\AddToHook{package/before/amsmath}{

}

\documentclass{article}

\usepackage{arxiv}

\usepackage{amsthm}
\usepackage{amsfonts}
\usepackage{amssymb}
\usepackage{amsmath}

\usepackage[utf8]{inputenc} 
\usepackage[T1]{fontenc}    
\usepackage{esvect}
\usepackage{epstopdf}		
\usepackage{hyperref}       
\usepackage{url}            
\usepackage{booktabs}       
\usepackage{nicefrac}       
\usepackage{graphicx}
\graphicspath{{./Figures/}}

\usepackage{doi}

\usepackage{comment}
\usepackage{dsfont}
\usepackage[symbol]{footmisc}

\usepackage{natbib}
\bibpunct[, ]{[}{]}{;}{a}{}{,}
\renewcommand\bibfont{\fontsize{10}{12}\selectfont}

\usepackage{algorithm}
\usepackage{algorithmicx}
\usepackage[noend]{algpseudocode}
\usepackage[labelsep=space]{caption}

\theoremstyle{plain}
\newtheorem{theorem}{Theorem}[section]
\newtheorem{lemma}[theorem]{Lemma}

\theoremstyle{definition}
\newtheorem{definition}[theorem]{Definition}

\theoremstyle{remark}

\usepackage{tikz}
\usetikzlibrary{shapes.geometric, arrows}
\tikzstyle{startstop} = [rectangle, rounded corners, minimum width=2.5cm, minimum height=1cm,text centered, draw=black, fill=white]
\tikzstyle{arrow} = [thick,->,>=stealth]

\title{Computing Volatility Surfaces using GANs with Minimal Arbitrage Violations}

\author{{Andrew S. Na}\thanks{CONTACT Andrew S. Na. Email: andrew.na@uwaterloo.ca}\hspace{1.5mm}\\
	David R. Cheriton School of Computer Science\\
	University of Waterloo\\
	Waterloo, ON\\
	\texttt{andrew.na@uwaterloo.ca} \\
	\And
	{Meixin Zhang} \\
	David R. Cheriton School of Computer Science\\
	University of Waterloo\\
	Waterloo, ON\\
	\texttt{meixin.zhang@uwaterloo.ca}
	\And
	{Justin W.L. Wan} \\
	David R. Cheriton School of Computer Science\\
	University of Waterloo\\
	Waterloo, ON\\
	\texttt{justin.wan@uwaterloo.ca}
}

\begin{document}

\maketitle

\begin{abstract}
In this paper, we propose a generative adversarial network (GAN) for efficiently computing volatility surfaces. Our framework trains a regularized generative adversarial network to compute volatility surfaces from time to maturity, moneyness and at the money implied volatility. We show an equivalent formulation between the GAN and the inverse problem. We incorporate calendar and butterfly spread arbitrage penalty terms to our generator loss functions that minimizes the arbitrage violations of our generated volatility surface. In this paper, we show that we can use GAN to speed up the computation of volatility surfaces while minimizing arbitrage violations. Our experiments show that by using regularization and the discriminator we can use a shallow network for the generator with accurate results. Comparing our models with other methods, we found that our method can outperform artificial neural network (ANN) frameworks in terms of errors and computation time. We show that our framework can generate market consistent implied and local volatility surfaces on minimal sample data by using a pre-trained discriminator and retraining the generator on market data.
\end{abstract}


\section{Introduction}
The Black-Scholes equation implies that the volatility parameter $\sigma$ is constant under geometric Brownian motion (GBM) with no market frictions \citep{cont-2002, lee-2002}. However, it is well documented that the observed market volatility is not constant \citep{lee-2002}. Black-Scholes implied volatility is one method to overcome this. Market option prices are often quoted in terms of Black-Scholes implied volatility. The Black-Scholes implied volatility $\sigma_{implied}(K,T)$ is solved from the Black-Scholes equation that matches the market price. The Black-Scholes implied volatility is often used by risk analysts and traders routinely for pricing derivatives and hedging risks \citep{cont-2002}. For easy exposition, we refer to the Black-Scholes implied volatility as the implied volatility.  

Another method to account for the non-constant volatility is to use stochastic volatility models to model asset prices \citep{lee-2002}. Stochastic volatility models assume the variance follows a stochastic process which is coupled with the asset price process. Stochastic volatility models are attractive to use as they can successfully model key features of volatility observed empirically such as curvature in large maturities. However, stochastic volatility models are difficult to calibrate due to the large number of parameters. In our paper, we use the Heston model to simulate asset prices. It is well known that the Heston model cannot replicate market prices for short expiries when parameters are time-consistent \citep{gatheral-2001-2}. This motivated the development of local volatility models. We let $\sigma_{local}(S,t)$ be the local volatility. Local volatility is computed as the solution to the inverse problem from the Dupires equation. The solution to the local volatility is not unique, which means the solution is difficult to verify \citep{gatheral-2001-1}. Absence of market data points also make it difficult to compute volatility surfaces that are consistent with market prices and great care must be taken during estimation \citep{coleman-2000, boyle-2000}. Also the lack of data availability in large and small strikes makes it more difficult to compute a volatility surface for all maturities and strikes. To avoid this we look at using synthetic market data computed from given Heston parameters. We assume the market prices $V_{mkt}$ generated from the Heston model is the true market price. The Heston model was chosen because it has a closed form solution for European options \citep{heston-1993}. 

The implied volatility from the Heston model can be computed using classical methods such as Brent's method or Newton's method at each maturity $T$ and strike $K$. At the calibrated $T$ and $K$ the implied volatility does not exhibit arbitrage. However, when we extrapolate and interpolate to another $T$ and $K$ arbitrage may be violated. Another challenge posed by SV models such as the Heston model is that the calibration of their parameters is difficult as the data may not be available in practice. Currently, practitioners use parametric models such as stochastic volatility inspired (SVI) to compute local volatility. The framework has extensions that ensure arbitrage is not violated in the local volatility surface and ensure that it is consistent with the Heston model in the limits \citep{gatheral-2011,gatheral-2014}.

Given market calibrated Heston parameters the implied volatility computed from the Heston model is arbitrage free at calibrated $T$ and $K$. However, no-arbitrage can be violated when we extrapolate or interpolate volatility points to other maturities and strikes. The computation of arbitrage free volatility surfaces is important because an arbitrage opportunity in the option market allows trading strategies that are implemented at zero cost and provides only upside potential. Such opportunities implies there are price mismatches in the market which may be exploited and cause a loss to the option writer or holder.

Recently, neural networks and machine learning models have been developed to compute the Black-Scholes implied volatility in real and synthetic markets \citep{liu-2019-2,hernandez-2016,poggio-2017,spiegeleer-2018,dimitroff-2018,horvath-2021,liu-2019-1,hirsa-2019}. Methods using convolutional neural networks (CNNs) has been explored in \citep{hernandez-2016} and \citep{dimitroff-2018}. They showed that CNNs performed well; however, they needed to be redesigned for specific models. Deep artificial neural networks (ANNs) such as the rough volatility method \citep{horvath-2021}, the implied volatility ANN (IV-ANN) method \citep{liu-2019-1}, the calibrating neural network (CaNN) method \citep{liu-2019-2}, and the method in \citep{hirsa-2019} are used to learn the solution of the option pricing function and construct the implied volatility surface by applying an inverse method or another neural network. The methods mentioned above do not account for the possible arbitrage violation of the neural network as there is no guarantee that the no-arbitrage conditions are not violated with generic neural networks. We also note that the inversion network of CaNN is model dependent as it uses the parameters of the model it is trying to calibrate as inputs used to learn the option price.

Neural network and machine learning models have been used in the computation of local volatility \citep{itkin-2019, chataigner-2021}. These methods typically account for arbitrage violations as arbitrage violation is a well known issue of local volatility, \citep{gatheral-2001-1}. The deep local volatility (DLV) method introduces a regularization term to the loss function to calibrate the implied volatility over an no-arbitrage option price surface. The method of \citep{itkin-2019} proposes an ANN with no-arbitrage constraints on learned parameters which guarantees that parameters calibrated on in-sample data adhere to no-arbitrage. Variational autoencoder (VAE) has been used in extrapolation and interpolation of local volatility surfaces \citep{bergeron-2022}. The VAE model was used to complete a local volatility surface where some of the local volatility points are known. The use of regularization terms to limit arbitrage violations is explored using VAEs using a similar approach to the DNN \citep{bergeron-2022}. A GAN model has been used in the calibration of local stochastic volatility (LSV) model parameters \citep{cuchiero-2020}. In comparison in this paper we look at computing volatility surfaces. The GAN calibration for LSV \citep{cuchiero-2020} also does not explore the use of regularization to enforce no-arbitrage conditions. We show later that including the no-arbitrage conditions as penalty terms has a noticeable affect on the volatility surface geometry. GAN for generating implied volatility with no-arbitrage penalty has been explored \citep{sidogi-2022}. However, the authors use the standard GAN loss function which results in irregular surfaces \citep{cont-2022}. We show in our numerical experiement that this results in volatility surfaces that are inconsistent with the market. VolGAN, a method to generate dynamic implied volatility and covariation from financial timeseries has been proposed \citep{cont-2022}. The authors also use the standard GAN loss with arbitrage penalties but use the log volatility to generate the surface which results in better performance. The inclusion of the mean-squared error (MSE) term in the loss function has been explored in the Fin-GAN framework \citep{vuletic-2023}.

Computing volatility surface from option prices often requires us to compute the option price. Though deep neural networks can successfully learn nonlinear functions, training deep networks requires a lot of data and can take a long time. In this paper, we propose using a generative adversarial network (GAN) to generate volatility surfaces from time to maturity, moneyness and at the money implied volatility. The generator network is assisted in training by a discriminator that evaluates whether the generated implied volatility matches the targets distribution or not. We derive a GAN framework from the inverse problem that is consistent with the target volatility using an MSE loss function \citep{vuletic-2023}. Our framework also trains our network to satisfy the no-arbitrage constraints by introducing penalties as regularization terms \citep{roper-2010, itkin-2019}. Although training a generator and a discriminator involves two networks, our proposed GAN model allows the use of shallow networks which results in much lower computational cost.

The contribution of this paper is as follows:
\begin{itemize}
\item We show that the use of a generator-discriminator pair allows us to train shallow networks to achieve greater efficiency without losing accuracy.
\item We propose a GAN based framework to compute volatility surfaces from synthetic market option prices generated by the Heston model. Our framework allows us to use the trained generator network to generate the implied volatility and the local volatility with minimal arbitrage violations out-of training, which is important for pricing and hedging options. 
\item We also show that our method can be used to generate market consistent volatility surfaces with minimal tuning of the generator on limited sample data by using a pre-trained discriminator.
\end{itemize}

\section{Heston Model, Volatility Surface and Static Arbitrage}\label{sec:background}
In the following section, we discuss the stochastic asset price model used to price the synthetic data used in this paper. Then we discuss implied volatility, local volatility and the computation of volatility surfaces. Finally, we discuss static arbitrage and present penalty terms that impose no-arbitrage as regularization on loss functions.

\subsection{Heston Model}
In this paper we model market asset prices using the Heston model. In standard Black-Scholes model, it assumes the asset price follows a geometric Brownian motion (GBM) which leads to the well known Black-Scholes equation for the no-arbitrage option price. However, it assumes the volatility is constant across maturities and strikes which may not hold in practice.

In this paper we use the Heston model to price the European call option which is used as synthetic data for our model training. We denote the call option price as $V$. The Heston model has the nice property that the characteristic function can be derived analytically which can be used in fast efficient solvers.  The asset price that follows the Heston model is given as the pair of SDEs \citep{heston-1993}
\begin{align*}
    dS_t &= r S_tdt +\sqrt{v_t}S_tdW_{t}^S \\
    dv_t &= \kappa(\bar{v} - v_0)dt +\gamma\sqrt{v_t}dW_{t}^v \\
    dW_{t}^S dW_{t}^v &=\rho dt,
\end{align*}
where $S_t$ is the stock price at time $t$, $r$ is the risk-free interest rate, $\kappa$ is the mean reversion rate, $\rho$ is the correlation between the stock price process and the variance process, $W_t^S$ is the Wiener process driving the stock price dynamics, $W_t^v$ is the Wiener process driving the variance process, $\gamma$ is the volatility of the variance, $\bar{v}$ is the long-run average of the variance, $v_t$ is the variance at time $t$ and $v_0$ is the initial variance.

We used the cosine (COS) method of \citep{fang-2009} to compute the European option price using the characteristic function of the Heston model. More precisely, let $i=\sqrt{-1}$, $\psi\in\mathbb{R}$ and let $\tau= T-t$ for any time $t\in[0,T]$.  The variables $\{\kappa,\rho,\gamma,v_0,\bar{v}\}$ are the Heston parameters. The explicit form of the characteristic function of the Heston model given by \citep{dunn-2014}
\[
    f(i\psi) = e^{A(\tau)+B(\tau)v_t+i\psi S_t}
\]
where
\begin{align*}
    A(\tau) &= ri\psi\tau + \frac{\kappa\bar{v}}{\gamma^2}\left(-(\rho\gamma i\psi-\kappa-M)\tau-2\ln\left(\frac{1-Ne^{M\tau}}{1-N}\right)\right)\\
    B(\tau) &= \frac{(e^{M\tau}-1)(\rho\gamma i \psi-\kappa-M)}{\gamma^2(1-e^{M\tau})}\\
    M &= \sqrt{(\rho\gamma i\psi + \kappa)^2+\gamma^2(i\psi + \psi^2)}\\
    N&= \frac{\rho\gamma i \psi -\kappa - M}{\rho\gamma i \psi -\kappa + M}.
\end{align*}

\subsection{Volatility Surface}

In this paper, we want to generate volatility surfaces through a trained generator network. This means we want our proposed model to generate the implied volatility and local volatility from our synthetic market prices.  In the following, we give an overview of the implied and local volatility and how they are computed.

\subsubsection{Implied Volatility}
Given the initial stock price $s_0$, the strike price $K$, the time to maturity $T$, the risk-free interest rate $r$, and the Heston parameters we can compute the arbitrage-free option price. Let $V_{mkt}(K,T)$ be the call option price at $T$ and $K$. In this paper we compute $V_{mkt}(K,T)$ using the Heston model. The implied volatility, $\sigma_{implied}(K,T)$, is found by solving the inverse problem
\begin{equation}
  V(S_0, K, T, r,\sigma_{implied}(K,T)) = V_{mkt}(K,T),
  \label{eq:calibration}
\end{equation}where the Black-Scholes call price, $V$, is found by solving the Black-Scholes equation \citep{liu-2019-1, gatheral-2001-1}. Note that the solution to the inverse problem, \eqref{eq:calibration}, can be reformulated as
\begin{equation}
V(S_0,K,T,r,\sigma_{implied}(K,T)) - V_{mkt}(K,T) =0.
\label{eq:root}
\end{equation}
Computing implied volatility is then reduced to the root finding problem given by \eqref{eq:root}. Two common root finding methods are Newton's method and Brent's method \citep{brent-1973}. In this paper, we will use Brent's method for solving \eqref{eq:root}.

\subsubsection{Local Volatility}
Given the stock price $S$, under local volatility the call option price follows the dynamics given by
\begin{equation*}
	\frac{\partial V}{\partial t} + \frac{1}{2}\sigma_{local}(S,t)^2S^2\frac{\partial^2 V}{\partial S^2}+rS\frac{\partial V}{\partial S} - rV = 0
\end{equation*}with terminal condition at $T$ is $V = \max(S-K,0)$. To change the dependent variables of the local volatility to $\sigma_{local}(K,T)$ we can form the equivalent forward equation called the Dupire equation \citep{dupire-1994}:
\begin{equation}
	\frac{\partial V}{\partial T} + \frac{1}{2}\sigma_{local}(K,T)^2K^2\frac{\partial^2 V}{\partial K^2}+rK\frac{\partial V}{\partial K} = 0
	\label{eq:dupire}
\end{equation}
with initial condition $V(K,0) = \max(S_0-K,0)$ \citep{lee-2002,gatheral-2001-1,gatheral-2001-2}. Let $\partial_T=\frac{\partial}{\partial T}$, $\partial_K=\frac{\partial}{\partial K}$, and $\partial_{KK}=\frac{\partial^2}{\partial K^2}$. Given the option price surface $V(K,T)$, in terms of $K$, and $T$ we can define the local volatility as
\begin{equation*}
	\sigma_{local}(K,T)=\left(\frac{\partial_TV+rK\partial_KV}{0.5K^2\partial_{KK}V}\right)^{1/2}.
\end{equation*}For a grid of $V_{mkt}$, $K$, and $T$ we can approximate  $\sigma_{local}$ from market option prices using finite-difference method (FDM). For discretized maturity $T_j$ with size $\Delta T$ and strike $K_i$ with step size $\Delta K$ the approximations to the partial derivatives are as follows
\begin{align*}
	\partial_T V_{i,j} &\approx \frac{V_{i,j}-V_{i,j-1}}{\Delta T}\\
    \partial_K V_{i,j} &\approx \frac{V_{i,j}-V_{i-1,j}}{\Delta K}\\
    \partial_{KK}V_{i,j} &\approx \frac{V_{i+1,j}-2V_{i,j}+V_{i-1,j}}{\Delta K^2},
\end{align*}then the discretized approximation of $\sigma_{local}(K,T)$ is given by
\begin{equation*}
	\sigma_{local}(K_i,T_j) \approx \sigma_{FDM}(K_i,T_j) = \left(\frac{\partial_TV_{i,j}+r_{i,j}K_{j}\partial_KV_{i,j}}{K_{j}^2\partial_{KK}V_{i,j}}\right)^{1/2}.
\end{equation*}Though the FDM method can be convenient, it has an issue with the regularity and availability of market data. To overcome the issue of limited market price data practitioners look at the parametric approximation given by the surface stochastic volatility inspired (SSVI) method \citep{gatheral-2011,gatheral-2014}.
\begin{definition}
	Let $\theta_T = \sigma_{implied}(s_0,T)^2T$ and let $\phi$ be a smooth function from $\mathbb{R}_{+} \mapsto\mathbb{R}_{+}$ such that the limit of $\theta_T\phi(\theta_T)$ exists as $T\rightarrow 0$. Then the SSVI is the surface defined by \citep{gatheral-2014}
    \begin{equation}
    	w(k,\theta_T) = \frac{\theta_T}{2}\left(1+\rho\phi(\theta_T)k+\sqrt{(\phi(\theta_T)+\rho)^2+(1-\rho^2)}\right).
    \label{eq:heston_like_ssvi}
    \end{equation}
\end{definition}
Note the parameter $\rho$ is the same as the Heston parameter $\rho$. To solve for the local volatility surface we let $\lambda > 0$ and use the Heston-like parametric function given by \citep{gatheral-2014}
\begin{equation}
	\phi(\theta) = \frac{1}{\lambda\theta}\left(1-\frac{1-e^{-\lambda\theta}}{\lambda\theta}\right).
\label{eq:heston-parametric-function}
\end{equation}

\subsection{Static Arbitrage}
The absence of static arbitrage in option prices ensures the fair valuation of the option for a fixed $T$. In this paper we refer to static arbitrage as arbitrage. We can say an option price surface or a volatility surface is free of arbitrage if and only if \citep{gatheral-2014}
\begin{enumerate}
\item it is free of calendar spread arbitrage (monotonicity);
\item it is free of butterfly spread arbitrage.
\end{enumerate}
The absence of calendar spread implies that monotonicity constraints of the option price/volatility surface is satisfied over $T$. The absence of butterfly arbitrage ensures that there exists a non-negative probability measure \citep{gatheral-2014}. This ensures that the option price/volatility is a martingale and arbitrage free.

Generally, arbitrage violation are found by checking that the option price/volatility surface violates one or both of these conditions \citep{carr-2005}, which are often imposed by constraints. In this paper, we distinguish two types of constraints. Hard constraints are constraints that must be satisfied to be a feasible solution. Alternatively, soft constraints are constraints that can be violated, but incur a penalty. This penalty term is then added as a regularization term to the objective function as shown in \citep{itkin-2019} and \citep{ackerer-2020}. Note that this method does not guarantee that all arbitrage opportunities will be eliminated.

There are two methods to impose soft constraints for no-arbitrage. The first method is to impose it on the option price surface directly \citep{itkin-2019}. The second method is to impose the soft constraints through the implied volatility surface \citep{ackerer-2020}.  We review both methods here, as the no-arbitrage on implied volatility surface builds on the work on price surfaces.

\subsubsection{No-Arbitrage Conditions imposed on Option Price Surface}
Let $V(K,T):[0,\infty)\times[0,\infty)\rightarrow\mathbb{R}$ be the call price surface. For $V(K,T)$ to be arbitrage free it must have the following properties \citep{carr-2005}:
\begin{equation}
    \partial_T V > 0,\partial_K V < 0, \partial_{KK}V > 0.
    \label{eq:narb-hard}
\end{equation}
The first constraint of \eqref{eq:narb-hard} ensures that calendar arbitrage (monotonicity) is not violated. The second and the third ensures that the butterfly arbitrage is not violated.

The constraints\eqref{eq:narb-hard} are applied to the option price surface to ensure that the option price surface is arbitrage-free \citep{carr-2005}. Neural networks are very difficult to train with hard constraints and tend to result in extremely large prediction errors \citep{ackerer-2020}. \citep{itkin-2019} treats the constraints as soft constraints and derives penalty functions to limit arbitrage violations. For given constants $\delta_1>0,\delta_2>0,\delta_3 > 0$ and the approximate option price, $\hat{V}$, the regularization terms are given by \citep{itkin-2019}:
\begin{equation}
    L_1 = \delta_1 \max\{-T^2\partial_T\hat{V},0\}, L_2 = \delta_2 \max\{-K^2\partial_{KK}\hat{V},0\}, L_3 = \delta_3 \max\{K\partial_K\hat{V},0\}.
    \label{eq:narb-soft}
\end{equation}

Another approach to ensure no-arbitrage via regularization is by treating the Dupire condition \citep{dupire-1994} as a soft constraint \citep{chataigner-2020}. Let $k = \log(K/s_0)$, $\partial_k = \frac{\partial}{\partial k}$ and $\partial_{kk} = \frac{\partial^2}{\partial k^2}$ then the Dupire penalty is given by:
\begin{equation*}
    L_{dup} = \frac{\partial_T V}{k^2\partial_{kk}V}.
\end{equation*}
It is clear that the Dupire penalty requires the conditions $\partial_T V \geq 0$ and $\partial_{kk} V > 0$. This may not hold true if relying on auto-differentiation. However, this may be remedied by using a similar treatment on the put option surface as \eqref{eq:narb-soft}. We remark that for the DLV method the local volatility is computed as $\sigma_{local}(k,T) = (L_{dup})^{1/2}$. Note that the soft constraints can not guarantee that the calibrated implied volatility will be arbitrage free as it is not guaranteed to completely eliminate arbitrage in the option price surface. 

\subsubsection{No-Arbitrage Conditions imposed on Volatility Surface}
Let the scaled volatility surface of $k$ be denoted by $I(k,T):\mathbb{R}\times[0,\infty)\mapsto[0,\infty]$. We can build on the no-arbitrage conditions on $V(k,T)$, and extend it on to $I(k,T) = \sigma(k,T)\sqrt{T}$. The sufficient conditions for $I(k,T)$ to be arbitrage free are \citep{roper-2010}
\begin{itemize}
    \item (Smoothness): $I(k,T)$ is twice differentiable for every $T>0$;
    \item (Positivity): for every $k\in\mathbb{R}$ and $T>0$, $I(k,T) > 0$;
    \item (Durrleman's Condition): for every $k\in\mathbb{R}$ and $T>0$, 
        \begin{equation}
            0\leq (1 - \frac{k\partial_kI}{I})^2-\frac{1}{4}(\partial_kI)^2+I\partial_{kk}I,
            \label{eq:DC}
        \end{equation}
        where $I = I(k,T)$;
    \item (Monotonicity in $T$): for every $k\in\mathbb{R}$, $I(k,\cdot)$ is non-decreasing;
    \item (Large moneyness behaviour): for every $T>0$, $\limsup\limits_{k\rightarrow \infty}\frac{I(k,T)}{\sqrt{2}k}\in[0,1)$; and 
    \item (Value at maturity): for every $k\in\mathbb{R}$, $I(k,0) = 0$.
\end{itemize}
Note that the butterfly spread no-arbitrage conditions are satisfied from \eqref{eq:DC} and calendar spread no-arbitrage conditions are satisfied by the monotonicity condition.

Extending this approach we can show the constraints to eliminate calendar and butterfly spread can be expressed as soft constraints when calibrating for the volatility surface $\sigma(k,T)$ \citep{ackerer-2020}. We let the total variance be defined as
\[
\omega(k, T) =\sigma^2(k,T) T.
\]
Let $\ell_{cal}$ be the risk in the total variance from the calendar spread arbitrage and $\ell_{but}$ be the risk in the total variance from the butterfly spread arbitrage, which is given by
\begin{align*}
\ell_{cal}(k,T) &=\partial_T \omega(k,T) \\
\ell_{but}(k,T) &= \left( 1- \frac{k \partial_k \omega (k,T)}{2 \omega(k,T))}\right)^2 -\frac{\partial_k \omega(k,T)}{4} \left( \frac{1}{\omega(k,T)}+\frac{1}{4} \right)+ \frac{\partial^2_{kk} \omega(k,T)}{2}.
\end{align*}The proof of why $\ell_{but} \geq 0$ ensures the butterfly spread arbitrage is not violated  is shown in \citep{gatheral-2014}.

Then we can define the penalty to satisfy the no-arbitrage conditions as \citep{ackerer-2020} 
\begin{itemize}
\item \textbf{(Monotonicity in $T$):}
\[
L_c = \frac{1}{M} \sum_{i} \max (0,-\ell_{cal}(k_i, T_i)),\ i=1,...,M
\]
\item \textbf{(Durrleman's condition):}
\[
L_{bf} = \frac{1}{M} \sum_{i} \max (0,-\ell_{but}(k_i, T_i)),\ i=1,...,M
\]
\item \textbf{(Large moneyness limit):}
\[
L_\infty = \frac{1}{M} \sum_{i} |\partial^2_{kk} \omega(k_i,T_i)|,\ i=1,...,M
\]
\end{itemize}
In our method we add these penalty terms to the objective function of the generator, which we show in Section \ref{sec:method}.

\section{GAN framework for Computing Volatility Surfaces}\label{sec:method}
In this section we present our proposed model used to compute volatility surfaces. Our proposed framework uses a generative-adversarial network (GAN) model to compute no-arbitrage volatility surfaces. The GAN model is composed of two neural networks, the generator network and the discriminator network. The generator learns to generate volatility surfaces that have minimal arbitrage violations. The discriminator network is trained with on volatility surfaces over different time to maturities and interest rates. The discriminator learns to classify the given data as true if the data is from the distribution of the volatility surface and false if it is not. Concurrently, the generator learns to generate data that is consistent with the distribution of the labels by minimizing detection from the discriminator network. This competition between the generator and discriminator allows the GAN to generate out-of-training samples that closely mimic the target in distribution \citep{goodfellow-2016}. Our GAN framework allows us to use smaller more efficient networks for the generator. Our proposed model can be used to compute both the local volatility and the implied volatility. We remark that our method is not dependent on the Heston model. To allow our model to train on limited number of observations, we engineer two additional features. We use an adjusted log-moneyness to enrich the feature space and the at the money (ATM) implied volatility as a form of target encoding.
\begin{figure}[htbp]
    \centering
    \includegraphics[width=0.9\textwidth]{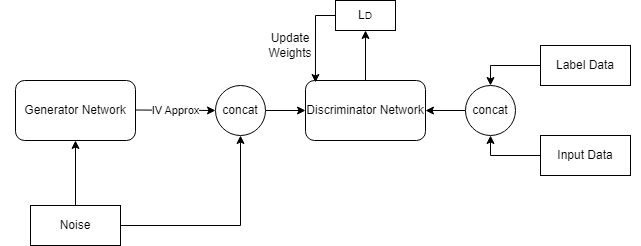}
    \caption{Schematic of discriminator network during training. The discriminator is trained with the Black-Scholes implied volatility (BS-IV) as the true labels and the generator output with noise as the fake labels. The generator weights are fixed while the discriminator is being trained.}
    \label{fig:GAN-disc}
\end{figure}

\begin{figure}[htbp]
    \centering
    \includegraphics[width=0.7\textwidth]{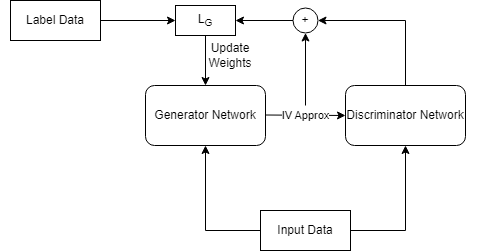}
    \caption{Schematic of generator network during training. The ground truth is given by the Black-Scholes implied volatility. The discriminator weights are fixed during the training of the generator network.}
    \label{fig:GAN}
\end{figure}

\subsection{GAN Formulation of the Inverse Problem}

We formally formulate the inverse problem that will be solved by our proposed GAN framework  to compute the complete implied and local volatility surfaces. Given a grid of strikes and maturities and the set of volatilities $\mathbb{Q}$, the objective function of the inverse problem requires us to solve:
\begin{equation}
    \arg\min_{\sigma\in\mathbb{Q}}\{\|V(S_0,K,T,r,\sigma(K,T))-V_{mkt}(K,T)\|_2^2\}=0.
    \label{eq:inverse_problem}
\end{equation}
\subsubsection{Implied volatility}
We formulate the inverse problem for the implied volatility directly, as it is given by definition. Given $\sigma_{mkt}(K_i,T_j)$ from $V_{mkt}(K_i,T_j)$, we can reformulate the inverse problem as:
\begin{equation}
    \min_{\sigma\in\mathbb{Q}}\{\|\sigma(K,T)-\sigma_{mkt}(K,T)\|_2^2\}=0.
    \label{eq:minmax-iv}
\end{equation}
To solve the implied volatility problem, it suffices to solve \eqref{eq:minmax-iv}. However, this may lead to values of $\sigma$ that may not be arbitrage free. Let $\mathbb{Q}_{implied} \subset \mathbb{Q}$ is the space of $\sigma(K,T)$ such that they are arbitrage free. By selecting $\sigma(K,T)$ such that it is most likely sampled from $\mathbb{Q}_{implied}$ we get the following minimax problem:
\begin{equation}
    \max_{\mathbb{Q}_{implied}\subset \mathbb{Q}}\min_{\sigma \in \mathbb{Q}_{implied}}\{\|\sigma(K,T) - \sigma_{mkt}(K,T)\|_2^2\} = 0,
    \label{eq:GAN-minmax-iv}
\end{equation}
we denote the resulting $\sigma^*(K,T)$ as $\sigma_{implied}(K,T)$.

\subsubsection{Local volatility}
To formulate the local volatility as an inverse problem in terms of volatility surfaces is not so clear. We present a Lemma that connects the local volatility surface to the inverse problem.
\begin{lemma}
Let $V_{mkt}$ be the market price of the option and $V$ its approximation. We assume that $\frac{\partial V_{mkt}}{\partial T}$, $\frac{\partial V_{mkt}}{\partial K}$, $\frac{\partial^2 V_{mkt}}{\partial K^2}$, $\frac{\partial V}{\partial T}$, $\frac{\partial V}{\partial K}$, $\frac{\partial^2 V}{\partial K^2}$ exists and $V$, $V_{mkt}$ satisfies \eqref{eq:dupire}. Under the assumption that there exists an optimal $\sigma^*(K,T)$ such that it satisfies \eqref{eq:inverse_problem}. Let $H(K,T,V,V_{mkt}) = \delta(V(K,T),V_{mkt}(K,T))$ and $\Gamma(K,T,\sigma,\sigma_{mkt}) = \delta(\sigma(K,T),\sigma_{mkt}(K,T))$, where $\delta(x,x_0) = \|x-x_0\|_2$ is the L2 distance between $x$ and $x_0$. Then we can form a controlled KFP equation given by measure $\delta$:
\begin{equation}
    \frac{\partial H}{\partial T}+2rK\frac{\partial H}{\partial K} +\frac{(V-V_{mkt})}{2}\min_{\sigma}\{\Gamma(K,T,\sigma)^2\}K^2\frac{\partial^2 H}{\partial K^2} = 0,
    \label{eq:kfp}
\end{equation}
with initial condition $H(K,0,V,V_{mkt}) = \delta(V(K,0),V_{mkt}(K,0))$. 
\end{lemma}
\begin{proof}
For $\delta = \|\cdot\|_2$, the distance between the market value and its approximation is given by $H=\delta(V, V_{mkt})$. We can construct the KFP equation. From chain rule
\begin{align*}
\frac{\partial H}{\partial T} &= \frac{\partial H}{\partial V}\frac{\partial V}{\partial T} + \frac{\partial H}{\partial V_{mkt}}\frac{\partial V_{mkt}}{\partial T}\\
&= \frac{\partial H}{\partial V}(-rK\frac{\partial V}{\partial K} - \frac{\sigma^2 K^2}{2}\frac{\partial^2 V}{\partial K^2}) +\frac{\partial H}{\partial V_{mkt}}(-rK\frac{\partial V_{mkt}}{\partial K} - \frac{\sigma_{mkt}^2 K^2}{2}\frac{\partial^2 V_{mkt}}{\partial K^2})\\
& = -2rK\frac{\partial H}{\partial K} - \frac{\sigma^2 K^2}{2}\frac{\partial H}{\partial V}\frac{\partial^2 V}{\partial K^2} - \frac{\sigma_{mkt}^2 K^2}{2}\frac{\partial H}{\partial V_{mkt}}\frac{\partial^2 V_{mkt}}{\partial K^2},
\end{align*}
which gives us the equation
\begin{equation}
    \frac{\partial H}{\partial T} + 2rK\frac{\partial H}{\partial K} + \frac{\sigma^2 K^2}{2}\frac{\partial H}{\partial V}\frac{\partial^2 V}{\partial K^2} + \frac{\sigma_{mkt}^2 K^2}{2}\frac{\partial H}{\partial V_{mkt}}\frac{\partial^2 V_{mkt}}{\partial K^2} = 0.
    \label{eq:intermediate_step}
\end{equation}
Next, we multiply and divide both sides of \eqref{eq:intermediate_step} by $H$. This gives us:
\begin{equation}
    \frac{\partial H}{\partial T} + 2rK\frac{\partial H}{\partial K} + \frac{\sigma^2 K^2}{2}(H)\frac{\partial H}{\partial V}(1/H)\frac{\partial^2 V}{\partial K^2} + \frac{\sigma_{mkt}^2 K^2}{2}(H)\frac{\partial H}{\partial V_{mkt}}(1/H)\frac{\partial^2 V_{mkt}}{\partial K^2} = 0,
    \label{eq:intermediate_step2}
\end{equation}
from the derivative of $H$ we can see that
\[
    \frac{\partial H}{\partial V} = \frac{V - V_{mkt}}{H} = -\frac{\partial H}{\partial V_{mkt}},
\]
\[
    \frac{\partial^2 H}{\partial V^2} = 1/H, \frac{\partial^2 H}{\partial V_{mkt}^2} = -1/H,
\]
applying the first derivative and the definition of the second derivative to \eqref{eq:intermediate_step2} we get
\[
    \frac{\partial H}{\partial T} + 2rK\frac{\partial H}{\partial K} + \frac{(V-V_{mkt}) K^2}{2}[\sigma^2 - \sigma_{mkt}^2]\frac{\partial^2 H}{\partial K^2} = 0,
\]
we can rewrite $\sigma^2 - \sigma_{mkt}^2$ as $\left(\sqrt{\sigma^2-\sigma_{mkt}^2}\right)^2 = \Gamma(\sigma,\sigma_{mkt})^2$. The initial condition can be derived as follows:
\begin{align*}
    \delta(V,V_{mkt}) &= \delta(\int_{T} \frac{\partial V}{\partial T}+rK\frac{\partial V}{\partial K} +\frac{1}{2}\sigma(K,T)^2K^2\frac{\partial^2 V}{\partial K^2}dT + V(K,0,\sigma)\\
    &,\int_{T} \frac{\partial V_{mkt}}{\partial T}+rK\frac{\partial V_{mkt}}{\partial K} +\frac{1}{2}\sigma_{mkt}(K,T)^2K_i^2\frac{\partial^2 V_{mkt}}{\partial K^2}dT + V_{mkt}(K,0,\sigma_{mkt}))\\
    &= \delta(V(K,0,\sigma),V_{mkt}(K,0,\sigma_{mkt})) = H(K,0,V,V_{mkt}).
\end{align*}
We know from the inverse problem that for each $T$ and $K$ we want 
\[
    H(V,V_{mkt}) \rightarrow 0\ \textnormal{as}\ \Gamma(\sigma,\sigma_{mkt}) \rightarrow 0,
\]
this tells us we need $\min_{\sigma}(\Gamma(\sigma,\sigma_{mkt})^2)$.
\end{proof}
Note, a connection exists between the forward Kolmogorov equation and the Hamilton-Jacobi-Bellman (HJB) equation \citep{annunziato-2014}. More, specifically, when the cost functional of the adjoint equation equals $d=\|\cdot\|_2$ and the $\sigma^*(K,T)$ is a strong solution, then the HJB equation is equal to the adjoint equation \citep{annunziato-2014}.

Let $\mathbb{Q}_{local} \subset \mathbb{Q}$ be space of volatilities $\sigma(K,T)$ such that they are arbitrage free and computed uniquely from Dupire's Equation. The optimal solution is found by taking the first order condition with respect to $\sigma$.
\[
    \min_{\sigma \in \mathbb{Q}}\{\|\sigma(K,T) - \sigma_{mkt}(K,T)\|_2^2\} = 0.
\]
By restricting the possible $\sigma$ to the space of arbitrage free volatilities, the inverse problem can be solved by solving the following problem:
\begin{equation}
    \max_{\mathbb{Q}_{local}\subset \mathbb{Q}}\min_{\sigma \in \mathbb{Q}_{local}}\{\|\sigma(K,T) - \sigma_{mkt}(K,T)\|_2^2\} = 0.
    \label{eq:GAN-minmax-loc}
\end{equation}
We denote the resulting $\sigma^*(K,T)$ as $\sigma_{local}(K,T)$. In our formulation, the $\max$ is solved using GAN which learns to accept solutions in $\mathbb{Q}_{implied}$ and $\mathbb{Q}_{local}$. The $\min$ is solved using the $MSE$ loss and no-arbitrage penalty terms given that $\sigma(K,T)$ belongs to the restricted set of volatilities $\mathbb{Q}_{implied}$ and $\mathbb{Q}_{local}$. We present the loss function explicitly in the next section. Note in practice, $\sigma_{mkt}$ is computed by finite differences.

\subsection{Loss Function and Proposed Model}

In this section we present our GAN model more formally. Our proposed model can be used for computing the implied and local volatility. For the implied volatility, the input of our proposed model is composed of: the risk free rate $r$, the moneyness $k = K/s_0$, the maturity time $T$ and the at the money (ATM) volatility $\sigma_{ATM} = \sigma_{implied}(s_0,T)$ and the adjusted log-moneyness $k_{log} = \log(k) - rT$. The adjusted log-moneyness is chosen to enrich the feature space by adding the log of the moneyness we account for skewness in the moneyness data, which is highly likely in volatility data as there may be different volatility levels given $T$. We also choose to include $\sigma_{ATM}$ as a form of target encoding. This is because under the limit at $T\rightarrow t$ the implied volatility converges to the spot volatility almost surely \citep{durrleman-2008}. Note this is a consistency condition for the model to be consistent with spot volatility \citep{durrleman-2008, carmona-2007}.

First, we construct the loss function to approximate the implied volatility. For each $T \in \mathbb{R}^{M}$ we draw $r$ from the range $[0.0,0.05]$. We let $k \in \mathbb{R}^{M}$ be the moneyness, $\sigma_{ATM}\in\mathbb{R}{+}^{M}$ be the ATM implied volatility and $k_{log} \in \mathbb{R}^{M}$ be the adjusted log-moneyness. We let $[X_1,...,X_N]^\top$ be the input to the GAN where $X_i = \{k,\sigma_{ATM},T,r\mathbf{1},k_{log}\}$, where $\mathbf{1}$ is a vector of ones of dimension $M$. Let $d = 5MN$ be the size of $X$. We let $b=MN$ be the size of the label data $y \in \mathbb{R}^{b}$ used in training, i.e. $\sigma_{implied}$ from Black-Scholes.

Let $G:\mathbb{R}^{d} \mapsto \mathbb{R}^{b}$ be the generator function and $D:\mathbb{R}^{d+b} \mapsto [0,1]^{b}$ be the discriminator function and let $z\sim N(0,1)$ such that $Z = [z_1,...,z_d]$. In a standard GAN model the learning objective is to ensure the distribution of $G(Z)$ is similar to the distribution of $y$. In our approach we do not use standard noise $Z$, but instead a shifted and scaled noise $\bar{Z} = [\bar{z}_1,...,\bar{z}_d]$, where $\bar{z}\sim N(\mu_X,\sigma_{X})$. The mean and standard deviation of X is estimated for each feature as the sample mean and standard deviation for the input $X$. This is done to ensure the noise is within the domain of  the input $X$. The standard loss function for the GAN is given by the minimax problem \citep{goodfellow-2016}
\begin{equation*}
	\min_{G}\max_{D} L(G,D) = \min_G\max_D \left\{\mathbb{E}[\log(D(\{X,y\}))] + \mathbb{E}[1-\log(D(\{\bar{Z},G(\bar{Z})\}))]\right\}
\end{equation*}
Solving the standard GAN loss function is difficult so in practice we reformulate the the standard loss function as two minimization problems given by
\begin{align*}
	\min_D L(D) &= \min_D  \{-\mathbb{E}[\log(D(\{X,y\}))]\}\\
    \min_G L(G) &= \min_G \{\mathbb{E}[1-\log(D(\{\bar{Z},G(\bar{Z})\}))]\}.
\end{align*}
These two loss functions are solved iteratively where the $G$ is fixed when minimizing $L(D)$ and $D$ is fixed when minimizing $L(G)$. 

We connect the minimax problems \eqref{eq:GAN-minmax-iv} and \eqref{eq:GAN-minmax-loc} to the standard GAN. The generator $G$ outputs possible values of $\sigma$ from a given input.  The discriminator $D$, classifies $\sigma$ as a volatility that is in $\mathbb{Q}_{implied}$, where $target = \{implied, local\}$
or not, this restricts the possible values of $\sigma$ such that they are consistent with $\mathbb{Q}_{target}$. The MSE loss is used to solve the inverse problem and map $\sigma$ to $\sigma_{target}$. Note that, we use lower case for scalar values, upper case for vector values and underlined upper case for matrices.

We formulate our generator loss function similarly to \citep{horvath-2021} and \citep{ackerer-2020} and use the mean squared error (MSE) to learn the shape of the volatility surface. For $b$ different samples the MSE is given by
\begin{equation}
    MSE(y, G(X)) = \frac{1}{b}\sum_{i = 1}^{b} (y_i - G(X_i))^2
\label{eq:rmse}
\end{equation}
To generate arbitrage-free implied volatility, we incorporate the no-arbitrage conditions $L_c, L_{bf}$ and $L_\infty$ in the loss function of the generator. To ensure that the $G(X)$ has a similar distribution to $y$ we minimize the negative log-likelihood loss function given by
\begin{equation*}
    L_{D_G} = -\frac{1}{b} \sum_{i = 1}^b \log(D(\{\bar{Z}_i,G(\bar{Z}_i)\})).
\end{equation*}
Then the loss function for the generator is given by
\begin{equation}
    L_G = MSE(y, G(X)) + \lambda_1 L_c + \lambda_2 L_{bf} + \lambda_3 L_\infty + \lambda_4 L_{D_G}.
\label{eq:generator_loss}
\end{equation}
The parameters $\lambda_1 > 0$, $\lambda_2 > 0$, and $\lambda_3 > 0$ are used to determine the amount of calendar arbitrage, butterfly arbitrage and the limit behaviour that are penalized. $\lambda_4 \in [0,1]$ is the amount of similarity we want with the distribution of $y$. Note that the additional terms act as regularizers to the standard MSE loss function. 

For the discriminator network, we use the binary cross-entropy (BCE) loss function \citep{goodfellow-2016} given by
\begin{equation}
    L_{D}= -\frac{1}{b} \sum_{i = 1}^b \left(\log(D(\{X_i,y_i\}))  + \log (1-D(\{X_i,G(X_i)\}))\right),
\label{eq:bce}
\end{equation}
The discriminator loss is chosen to maximize the likelihood that the discriminator classifies the target values as true given some noisy inputs. Then we want to minimize $L_D$ and $L_G$ iteratively such that $D^* = \min_D L_D$ and $G^* = \min_G L_G$. 

For the local volatility the loss function is the same except the input of our proposed model is given by $X_i =\{k,\sigma_{ATM},\sigma_{implied},T,r\mathbf{1},k_{log}\}$ and the label $y = \sigma_{local}$.

In our framework we model $G$ and $D$ using feedforward neural networks. We define the generator network as $\mathcal{G}(\cdot;\Omega)$ with a set of parameters $\Omega$. The discriminator network as $\mathcal{D}(\cdot;\Theta)$ with a set of parameters $\Theta$.  The generator network we use is the ANN with $l=1,2$ layers. We did not use a deep neural network as other works on volatility computation such as \citep{liu-2019-1} and \citep{chataigner-2020}. We found that increasing the depth past $l=2$ did not affect the accuracy of the generator network. We remark that depth past $l=1$ in the discriminator network actually lead to a degradation of performance as it was over-fitting on synthetic data. inputs of the network are normalized by the mean and standard deviation. Let $W_l$ be the weights of the $l$-th layer neural network. Then $\mathcal{D}(\{X,y\};\Theta)$ with parameters $\Theta = [W_0,W_1]$ is represented by
\begin{align*}
	&Z_1 = \textnormal{softplus}(\textnormal{batchnorm}(\{X,y\} W_0);\beta)\\
    &\mathcal{D}(\{X,y\};\Theta) = \textnormal{sigmoid}(Z_1 W_1).
\end{align*}
We represent GAN-1 with a $l=1$ layer $\mathcal{G}_1(X;\Omega_1)$, with parameters $\Omega_1 = [W_0,W_1]$ by
\begin{align*}
	&Z_1 = \textnormal{softplus}(\textnormal{batchnorm}(X W_0);\beta)\\
    &\mathcal{G}_1(X;\Omega_1) = \textnormal{softplus}(Z_1 W_1;\beta),
\end{align*}
and represent GAN-2 with a $l=2$ layer $\mathcal{G}_2(X;\Omega_2)$, with parameters $\Omega_2 = [W_0,W_1,W_2]$ as
\begin{align*}
	&Z_1 = \textnormal{softplus}(\textnormal{batchnorm}((X W_0);\beta)\\
    &Z_2 = \textnormal{softplus}(\textnormal{batchnorm}((Z_1W_1);\beta)\\
    &\mathcal{G}_2(X;\Omega_2) = \textnormal{softplus}(Z_2 W_2;\beta).
\end{align*}
We use a scaled version of the softplus activation function and sigmoid activation function defined as \citep{dugas-2001, goodfellow-2016}
\begin{align*}
    \textnormal{softplus}(X;\beta) &= \frac{1}{\beta}\log(1+e^{\beta X}),\\
    \textnormal{sigmoid}(X) &= \frac{1}{1+e^{-X}}.
\end{align*}
The softplus activation function is used for the generator network output because we want a smooth enough function to learn no-arbitrage soft constraints \citep{dugas-2001}. The sigmoid activation function for the discriminator output is standard for classification problems \citep{goodfellow-2016}.  We also employ batch normalization in each layer with learnable parameters $\gamma$ and $\eta$. The input $X$ is batch normalized as follows \citep{goodfellow-2016}:
\[
    batchnorm(X) = \gamma\frac{X-\mathbb{E}[X]}{std[X]} + \eta.
\]
We train the discriminator network first using $G(\bar{Z};\Omega)$ with $\Omega$ fixed as shown in figure \ref{fig:GAN-disc}. Then we train the generator network $\mathcal{G}(X;\Omega)$ using $\mathcal{D}(\{X,\mathcal{G}(X;\Omega)\};\Theta)$ with $\Theta$ fixed as shown in figure \ref{fig:GAN}. Reparameterizing \eqref{eq:generator_loss} and \eqref{eq:bce} by the neural network parameters $\{\Omega, \Theta\}$ gives us our final  loss function for the discriminator as
\begin{equation}
	\mathcal{L}_\mathcal{D}(\Theta) = -\frac{1}{b} \sum_{i = 1}^b \left(\log(\mathcal{D}(\{X_i,y_i\};\Theta))  + \log (1-\mathcal{D}(\{X_i,G(X_i)\};\Theta))\right),
 \label{eq:d_net_loss}
\end{equation}
 and for the generator as
 \begin{equation}
	\mathcal{L}_{\mathcal{G}}(\Omega) = MSE(y, \mathcal{G}(X;\Omega)) + \lambda_1 L_c + \lambda_2 L_{bf} + \lambda_3 L_\infty + \lambda_4 L_{\mathcal{D}_\mathcal{G}}(\Omega),
 \label{eq:g_net_loss}
\end{equation}
where 
\[
	L_{\mathcal{D}_\mathcal{G}}(\Omega) = -\frac{1}{b} \sum_{i = 1}^b \log(\mathcal{D}(\{\bar{Z}_i,\mathcal{G}(\bar{Z}_i;\Omega)\})).
 \]
The optimal set of discriminator network parameters is denoted by $\Theta^*$ and is found by
\begin{equation*}
    \Theta^* = \textnormal{arg}\min_{\Theta}\{\mathcal{L}_D(\Theta)\}.
\end{equation*}
The optimal set of generator network parameters is denoted by $\Omega^*$, which is found by minimizing \eqref{eq:generator_loss} this gives us
\[
    \Omega^* = \textnormal{arg}\min_{\Omega}\{\mathcal{L}_G(\Omega)\}
\]

\subsection{Training and Summary of GAN}
Training the GAN requires two steps. The first step is to train the discriminator. However, loss \eqref{eq:d_net_loss} is not straightforward to compute directly. Instead, the discriminator is trained using  algorithm \ref{alg:d_net_train}.
\begin{algorithm}[htpb]
\caption{Discriminator training algorithm}\label{alg:d_net_train}
\begin{algorithmic}
\Require $X,y,\mathcal{G}(\Omega^*), optimizer, epoch$
\State $n \gets 0$
\While{$n < epoch$}
\For{batch}
    \State $D := \mathcal{D}(\Theta)$
    \State $G:=\mathcal{G}(\Omega^*)$
    \State $\mathcal{L}_{\mathcal{D}_1}(\Theta) = -\mathbb{E}[\log(D(\{X,y\}))]$
    \State $\mathcal{L}_{\mathcal{D}_2}(\Theta) = -\mathbb{E}[\log(1- D(\{X,G(X)\}))]$
    \State $\mathcal{L}_{\mathcal{D}}(\Theta) = \mathcal{L}_{\mathcal{D}_1}(\Theta) + \mathcal{L}_{\mathcal{D}_2}(\Theta)$
    \State propogate errors backwards through the network.
    \State $optimizer$ step
\EndFor
\State $n \gets n + 1$
\EndWhile
\end{algorithmic}
\end{algorithm}
The total loss function $\mathcal{L}_\mathcal{D}(\Theta) = \mathcal{L}_{\mathcal{D}_1}(\Theta) + \mathcal{L}_{\mathcal{D}_2}(\Theta)$ is minimized using the Adam \citep{kingma-2014} optimizer.

In the second step, given $X$ and $y$ we train the generator by minimizing \eqref{eq:g_net_loss} using Adam. This is done in a standard neural network approach. Training is performed for $50$ epochs. Once training is complete, we evaluate the GAN with a forward pass using the testing data. 

A detailed summary of our training pipeline is summarized as follows
\begin{enumerate}
\item The Heston parameters are generated using a uniform distribution over a range of values as shown in table \ref{tab:data_gen} and the option price and target volatility surface $Y$ is computed as shown in figure \ref{fig:data_gen_pipeline}. Input $X$ and $Z$ is constructed for $M$ samples. 
\item Then $X, Z$ and $y$ are given as inputs to $\mathcal{D}(\cdot;\Theta)$.
\item $\mathcal{D}(\cdot;\Theta)$ is trained by evaluating \eqref{eq:d_net_loss} and $\mathcal{G}(X;\Omega)$ is trained by evaluating \eqref{eq:g_net_loss}.
\item The trained network $\mathcal{G}(X;\Omega)$ approximates the volatility surface.
\item We carry forward the weights of the generator from the previous epoch to initialize the weights of the next epoch to reinforce the soft constraints.
\end{enumerate}
The architecture and training parameters are detailed in table \ref{tab:model_info}.
\begin{table}[!htbp]
\centering
\caption{Model Parameters used in GAN for GAN-1 and GAN-2.}
\begin{tabular}{lc}
\hline
\textbf{Parameters} & \textbf{Options}  \\
\hline
Neurons(each layer) & $100$ \\
Activation function & softplus ($\beta=1$), sigmoid\\
Dropout rate & $0.0$ \\
Batch-normalization & No \\
Optimizer & Adam \citep{kingma-2014} \\
Batch size & 128 \\
\hline
\end{tabular}
\label{tab:model_info}
\end{table}
\subsection{Generating Data}\label{sec:data}
Our proposed model can be used to compute volatility surfaces, which are two separate problems thus they require different input features for training and testing. In this paper, our option price is generated using the Heston model with different combination of parameters. We use synthetic data over real data as it is difficult to obtain good quality real data for all $T$ and $K$. In this section we provide more details on the training and testing data used in this paper. In this paper we generate three separate sets of data. One set used for training and testing of our model. A second set for generating out-of training volatility surface and a third set for testing price errors.

\subsubsection{Generating Volatility Data}
We generate synthetic market data for training our GAN framework. The characteristic function of the Heston model was used with the COS method to generate the European call price, of an underlying asset following the geometric Brownian motion. Then the Black-Scholes implied volatility was computed using Brent's method. The finite-difference method was used to approximate the local volatility from European call prices.

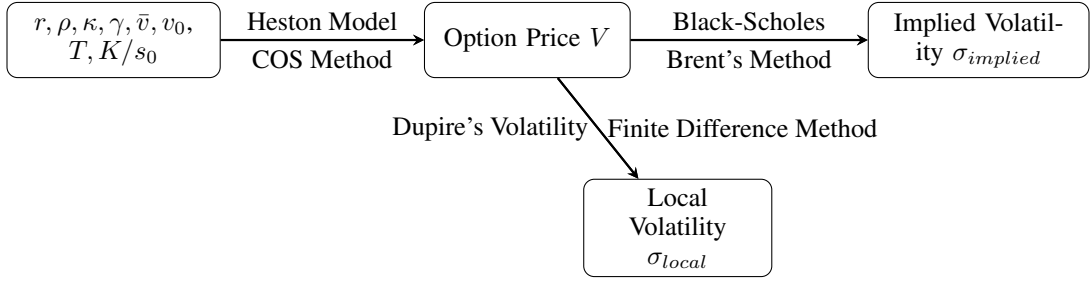
\begin{figure}[htbp]
\begin{center}
\begin{tikzpicture}[node distance=2cm]
    \node (inputs) [startstop, text width=2.6cm] {$r, \rho, \kappa, \gamma, \bar{v}, v_0$, $T, K/s_0$};
    \node (result1) [startstop, text width=2.5cm, right of=inputs, xshift=3.5cm] {Option Price $V$};
    \node (result2) [startstop, text width=2.7cm, right of=result1, xshift=4cm] {Implied Volatility $\sigma_{implied}$};
    \node (result3)[startstop, text width= 2.2cm, right of=result1, yshift=-2.5cm]{Local Volatility $\sigma_{local}$};
    
    \draw [arrow] (inputs) -- node[anchor=south] {Heston Model} (result1);
    \draw [arrow] (inputs) -- node[anchor=north] {COS Method} (result1);
    \draw [arrow] (result1) -- node[anchor=south] {Black-Scholes} (result2);
    \draw [arrow] (result1) -- node[anchor=north] {Brent's Method} (result2);
    \draw [arrow] (result1) -- node[below, midway, anchor=east] {Dupire's Volatility} (result3);
    \draw [arrow] (result1) -- node[below, midway, anchor=west] {Finite Difference Method} (result3);
\end{tikzpicture}
\end{center}
\caption{Pipeline used to generate data for training and testing our model.}
\label{fig:data_gen_pipeline}
\end{figure}

For training and validating, we use the parameters shown in table \ref{tab:data_gen} in our simulations. Each parameter was sampled randomly from the uniformly distributed parameters.

\begin{table}[htbp]
\centering
\caption{Domain of parameters used in the simulation of training and testing data, all parameters were drawn randomly from the parameter space.}
\begin{tabular}{l|c}
\hline
\textbf{Parameter} & \textbf{Range}\\
\hline
$r$, risk-free interest rate & $(0.0, 0.05)$\\
$\kappa$, reversion speed & $(0.0, 3.0)$\\
$\rho$, correlation & $(-0.9, 0.0)$\\
$\gamma$, volatility of variance & $(0.01, 0.5)$\\
$\bar{v}$, long-run mean variance & $(0.01, 0.5)$\\
$v_0$, initial variance & $(0.05, 0.5)$\\
$K/s_0$, moneyness & $(0.5, 2.5)$\\
$T$, time to maturity & $(0.5, 2.0)$\\
\hline
\end{tabular}
\label{tab:data_gen}
\end{table}
For training and validating, we generated $10$ different combinations of parameters above. For each set of parameters, we generated the option price and implied volatility for $75$ different maturity dates, and $50$ different strikes. 

For testing out-of training volatility surface we use $1$ set of parameters over $11$ different maturities and $157$ different strikes with parameters shown in table \ref{tab:data_gen_out}. For price error we use $50$ different sets of parameters, $8$ different maturities and $11$ different strikes with parameters shown in table \ref{tab:data_gen}.

\begin{table}[htbp]
\centering
\caption{Domain of parameters used in out-of training volatility surface generation, all parameters were drawn uniformly from the domain.}
\begin{tabular}{l|c}
\hline
\textbf{Parameter} & \textbf{Range}\\
\hline
$r$, risk-free interest rate & $0.02$\\
$\kappa$, reversion speed & $2.7$\\
$\rho$, correlation &$ -0.4$\\
$\gamma$, volatility of variance & $0.2$\\
$\bar{v}$, long-run mean variance & $0.4$\\
$v_0$, initial variance & $0.4$\\
$K/s_0$, moneyness & $(0.3, 2.8)$\\
$T$, time to maturity & $(0.3,2.0)$\\
\hline
\end{tabular}
\label{tab:data_gen_out}
\end{table}

\subsubsection{Volatility Data From Call options on the S\&P 500}
We also test our GAN framework on market index options built on the S\&P 500 index. The options dataset has the following features. The strike, moneyness, bid price, mid price, ask price, last price, volume, implied volatility and time traded. We gathered data for $T = 0.25, 0.5, 0.75, 1$. The interest rate $r$ was set to the $1$-year treasury bond rate of $5.463\%$. We collected a total of $821$ data points of moneyness $\in [0.24,1.17]$ and time-to-maturity $\in \{0.25, 0.5, 0.75, 1.0\}$.

\section{Numerical Results}\label{sec:Results}
In this section, we present our experimental results. Our experiment is divided into two main categories. First we show the that our proposed method can compute the Heston implied volatility surface. To measure the accuracy of our method we use the absolute percent error (MAPE) as a measure of performance. The MAPE is given by
\begin{equation*}
    MAPE = \frac{1}{b} \sum_{i = 1}^b \frac{|\sigma_{BS,i} - \sigma_{implied,i}|}{\sigma_{BS,i}}.
\end{equation*}
We also evaluate another measure of performance, the mean absolute error(MAE), which is given by
\begin{equation*}
    MAE = \frac{1}{b} \sum_{i = 1}^b |\sigma_{BS,i} - \sigma_{implied,i}|.
\end{equation*}Next we show that our proposed method can compute the local volatility surface. Note that the local volatility surface is not unique thus it is difficult to measure the quality directly. Instead we use repricing error as a method to measure the quality of our method \citep{horvath-2021,chataigner-2020,chataigner-2021}. This done by computing the implied volatility from the local volatility since
\[
    \hat{\sigma} = \sqrt{\sigma_{local}/T}.
\]
To measure the quality of our method we compute the average relative pricing error (ARPE) of the European option given by
\begin{equation*}
	ARPE = \frac{1}{b}\sum_{i=1}^{b}\frac{|V_{mkt,i}-V_{local,i}|}{|V_{mkt,i}|},
\end{equation*}
where  $V_{local}$ is the Black-Scholes European call option price computed using the local volatility. We use two other metrics to measure the performance of our proposed method for local volatility. We use the maximum relative price error (MRPE) given by
\begin{equation*}
	MRPE = \max_{i=1,...,b}\frac{|V_{mkt,i}-V_{local,i}|}{|V_{mkt,i}|},
\end{equation*}
and the standard deviation of relative price error. Note that all three measures are presented as error heatmaps \citep{horvath-2021}. In all of our experiments, we use generated data as described in Section \ref{sec:method}. In experiment 1, we compare the performance of GAN-1 vs. GAN-2 and show qualitative results which compares the output of our method to the implied volatility smile generated using Brent's method. We also compare our model with and without soft constraints. In experiment 2, we compare our proposed methods with the IV-ANN method of \citep{liu-2019-1}, we chose not to compare with the deep calibration \citep{horvath-2021} approach as the network architecture used in this model is captured by the IV-ANN method. We compute the implied volatility and present a qualitative and quantitative analysis. We also look at the repricing error of our method and the IV-ANN method to replicate $V_{mkt}$. In experiment 3, we compare our proposed method with a deep neural network implementation similar to the DNN \citep{chataigner-2020}. We compute the local volatility using both methods and compare the qualitative and quantitative results of both methods. In experiment 4 we  compare our method to a VAE implementation. In this experiment we compare the local volatility of both methods given a set computation time and evaluate the performance of both methods. Finally, in experiment 5 we use our pre-trained discriminator and fine-tune a 2-layer generator to generate market consistent volatility surfaces, this highlights our models capabilities to generalize to other datasets with minimal tuning with small number of samples (less than 1000). All our numerical experiments were run using Google Colab with 13 GB of RAM and a dual-core CPU of 2.2 GHz.

\subsection{Experiment 1: Performance Comparison between GAN models}

In this experiment, we compare the performance of the two GAN models proposed in our paper. We use MAE and MAPE to measure the performance of both models for implied volatility. For the local volatility we use the ARPE, MRPE and the standard deviation of repricing error heatmaps. To begin our comparison, we compare the training performance of GAN-1 and GAN-2. The training procedures for GAN-1 and GAN-2 are identical and are described in section \ref{sec:method}. 

\subsubsection{Implied Volatility Smile}
To show the effectiveness of our approach, we present the MAE and MAPE of the generator, which is a measure of how well our generated implied volatility matches the training and validation true implied volatility. We see that the training and testing loss for both of our models behave similarly for the generator. We summarize the training performance of GAN-1 and GAN-2 in table {\ref{tab:final_model_performance}}, which shows the MAE and MAPE of both models and their respective training times. We notice that GAN-1 trains faster than GAN-2 by $10$ seconds but GAN-2 tends to yield better accuracy than GAN-1 in both MAE and MAPE. However, only looking at error does not give us a full picture of the performance of our proposed method. Thus we also compare more qualitative results.

\begin{table}[htbp]
\centering
\caption{We compare the performance of GAN-1 and GAN-2 using training time, MAE and MAPE.}
\begin{tabular}{lccc}
\hline
\textbf{Model} & \textbf{Training Time} & \textbf{MAE} & \textbf{MAPE} \\
\hline
GAN-1 & 198.689732$s$ & $4.2680e^{-5}$ & $0.07898\%$\\
GAN-2 & 208.098838$s$ & $2.1376e^{-5}$ & $0.03981\%$\\
\hline
\end{tabular}
\label{tab:final_model_performance}
\end{table}

We compare the qualitative performance of our proposed models GAN-1 (dashed red) and GAN-2 (dashed green) for different maturities as shown in figure \ref{fig:comp_models_fit}. We show a cross section of the implied volatility surface for different maturities to highlight the tail behaviour of our generated surface. Our target value shown by the solid blue line was constructed from the implied volatility. The red dashed lines are the implied volatilities output from GAN-1 and the green dashed lines are the implied volatilities output from GAN-2. We can see that the implied volatility approximated by GAN-2 follows the implied volatility curve but does not match it completely. The implied volatility approximated by GAN-1 fails to capture the curvature of the implied volatility curve in maturities further out.  Though GAN is a powerful tool the generator network still requires some hidden layers to perform well. However, the training time required for GAN-2 is only $10$ seconds more than GAN-1 showing that it is still efficient.

\begin{figure}[htbp]
\centering
\includegraphics[width=\textwidth]{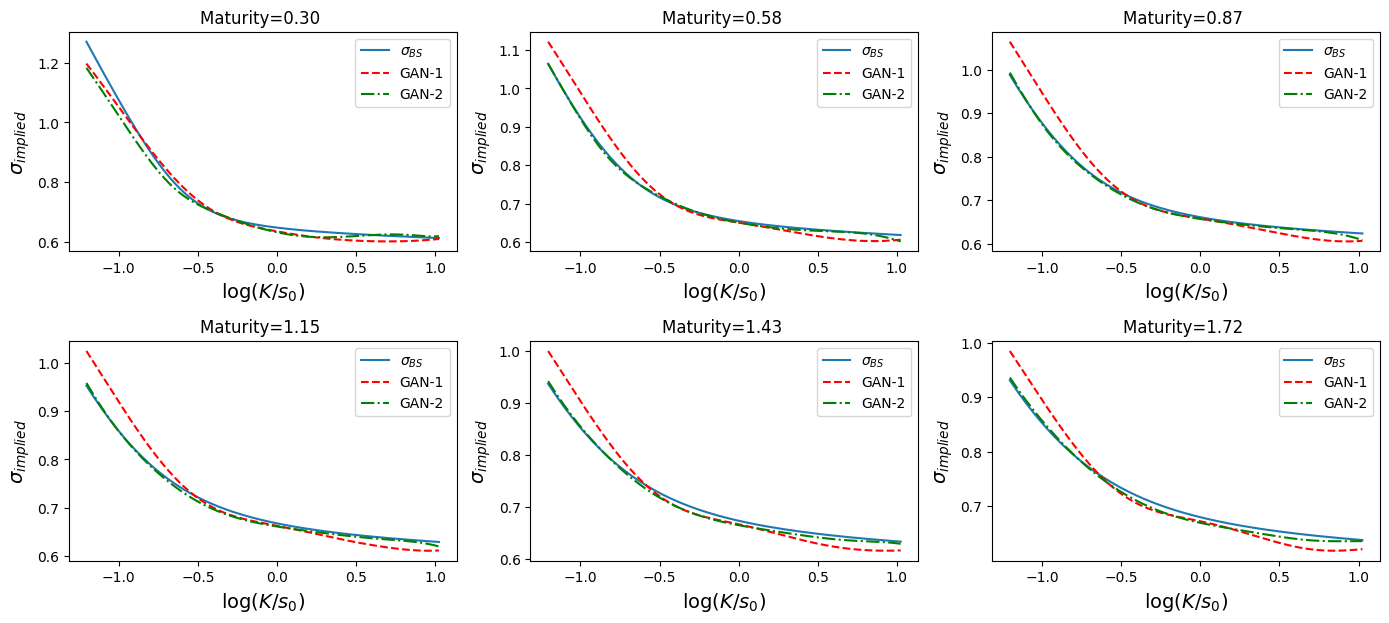}
\caption{The implied volatility computed by GAN-1 with soft constraints (red) and GAN-2 with soft constraints 
 (green) compared to the Black-Scholes implied volatility (blue).}
\label{fig:comp_models_fit}
\end{figure}

We compare the generated implied volatility from our models with soft constraints (dashed green) vs our models with out any constraints (dashed red) as shown in figures \ref{fig:fit_constraintsvnone} and \ref{fig:fit_constraintsvnone_gan2}.  As with figure \ref{fig:comp_models_fit}, we look at the cross section of the implied volatility surface at different maturities. We observe that the role of the regularizer as maturity increases. We see that the no-arbitrage penalty terms flatten the curve and allow for better fitting with the implied volatility curve. In the case with no constraints we can qualitatively observe some concavity in the implied volatility curves which suggests that arbitrage has been violated.

\begin{figure}[htbp]
\centering
\includegraphics[width=\textwidth]{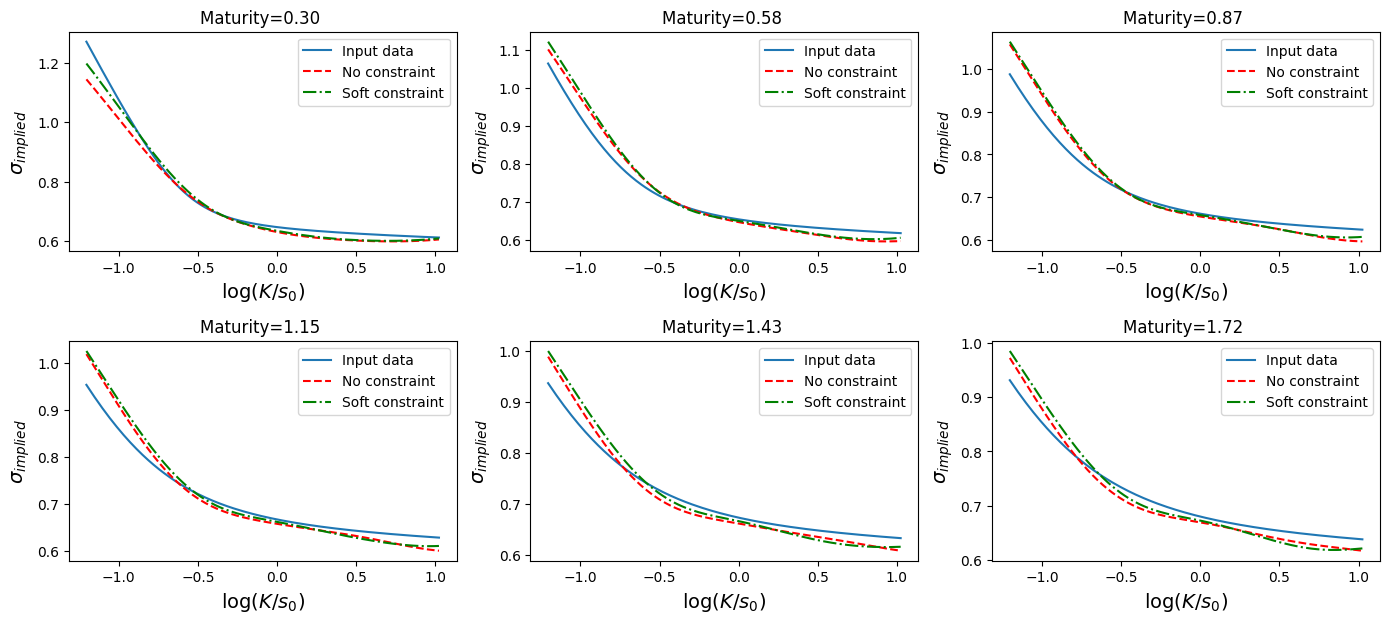}
\caption{The implied volatility approximated by GAN-1 with no constraints (orange) and GAN-1 with soft constraints (green) compared to the Black-Scholes implied volatility (blue).}
\label{fig:fit_constraintsvnone}
\end{figure}

\begin{figure}[htbp]
\centering
\includegraphics[width=\textwidth]{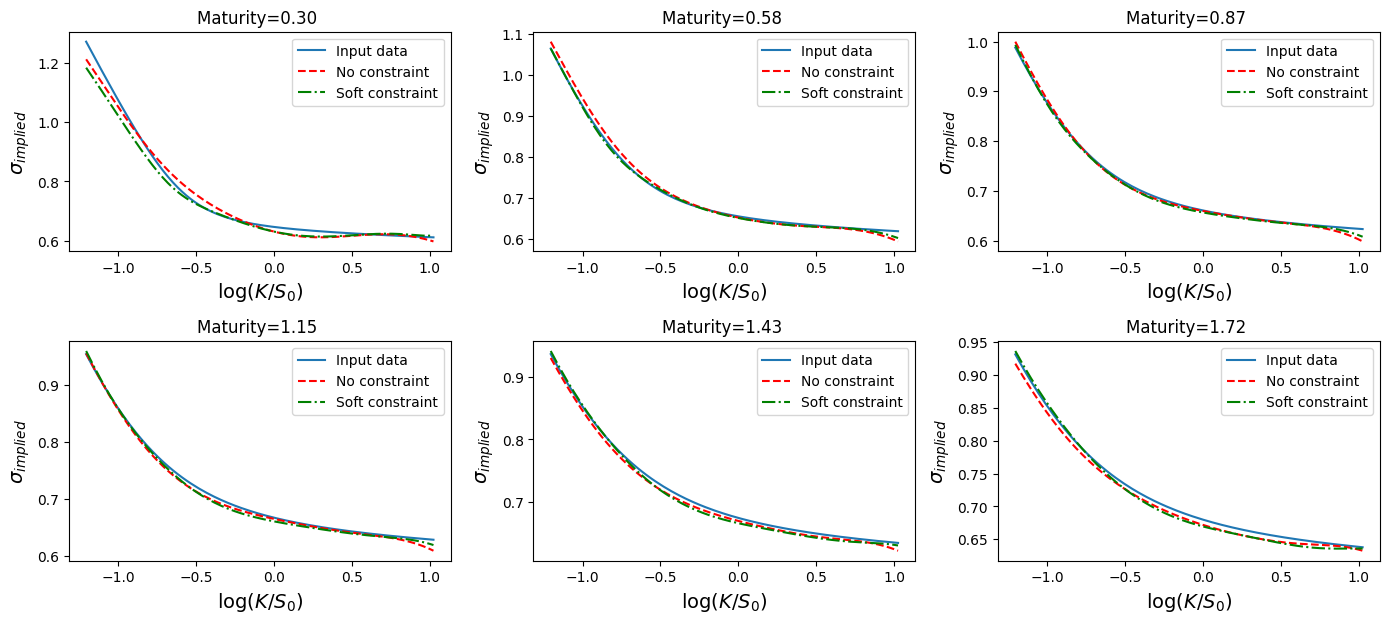}
\caption{The implied volatility approximated by GAN-2 with no constraints (orange) and GAN-2 with soft constraints (green) compared to the Black-Scholes implied volatility (blue).}
\label{fig:fit_constraintsvnone_gan2}
\end{figure}

We show the number of arbitrage violations during training and testing for both GAN models when soft constraints are used and when they are not used in tables \ref{tab:butterfly_arbitrage_violations} and \ref{tab:calendar_arbitrage_violations}. We found that adding the soft constraints greatly assisted the generator in avoiding arbitrage butterfly and vertical spread violations. From testing, we see that GAN-1 and GAN-2 without soft constraints performs considerably worse, as we had $4.90\%$ and $11.15\%$ arbitrage violations respectively. From table \ref{tab:butterfly_arbitrage_violations} we see a reduction in butterfly and vertical spread arbitrage violations when soft constraints are used. Note that the soft constraint does not guarantee that arbitrage violations do not occur.
\begin{table}[htbp]
\centering
\caption{Number of butterfly arbitrage violations detected in implied volatility surface.}
\begin{tabular}{lcccc}
\hline
\multicolumn{3}{c}{\textbf{GAN-1}} & \multicolumn{2}{c}{\textbf{GAN-2}} \\
\hline
& Soft Constraint & No Constraint & Soft Constraint & No Constraint \\
\hline
Training & 0/28306 & 1049/28306 & 0/28306 & 2143/28306 \\
Testing & 0/4995 & 245/4995 & 0/4995 & 557/4995 \\
\hline
\end{tabular}
\label{tab:butterfly_arbitrage_violations}
\end{table}

\begin{table}[htbp]
\centering
\caption{Number of calendar arbitrage violations detected in implied volatility surface.}
\begin{tabular}{lcccc}
\hline
\multicolumn{3}{c}{\textbf{GAN-1}} & \multicolumn{2}{c}{\textbf{GAN-2}} \\
\hline
& Soft Constraint & No Constraint & Soft Constraint & No Constraint \\
\hline
Training & 0/28306 & 0/28306 & 0/28306 & 0/28306 \\
Testing & 0/4995 & 0/4995 & 0/4995 & 0/4995 \\
\hline
\end{tabular}
\label{tab:calendar_arbitrage_violations}
\end{table}

\subsubsection{Local Volatility Smile}
Our proposed methods can learn the behaviour of the full local volatility with minimal arbitrage violations. We demonstrate the resulting local volatility surface computed by GAN-1 with soft constraints (left), GAN-2 with soft constraints (middle) and the FDM local volatility (right) in figure \ref{fig:comp_models_loc_vol}. This is done to show qualitatively how the generated local volatility surfaces compare to the approximate local volatility given by FDM. The surface generate by GAN-2 tends to try to fit the FDM local volatility more closely than GAN-1.
\begin{figure}[htbp]
\centering
\includegraphics[width=0.32\textwidth]{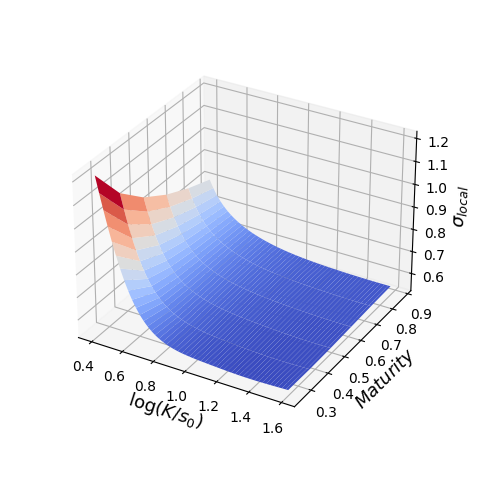}
\includegraphics[width=0.32\textwidth]{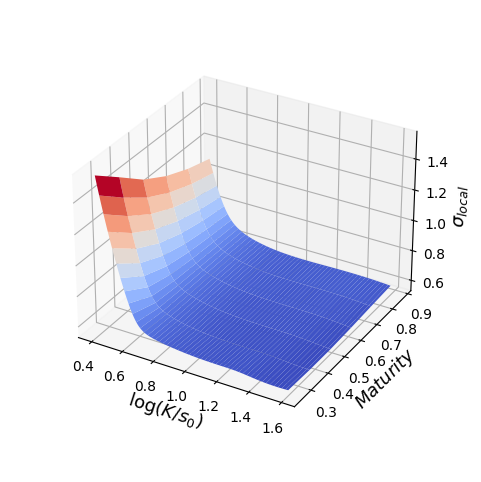}
\includegraphics[width=0.32\textwidth]{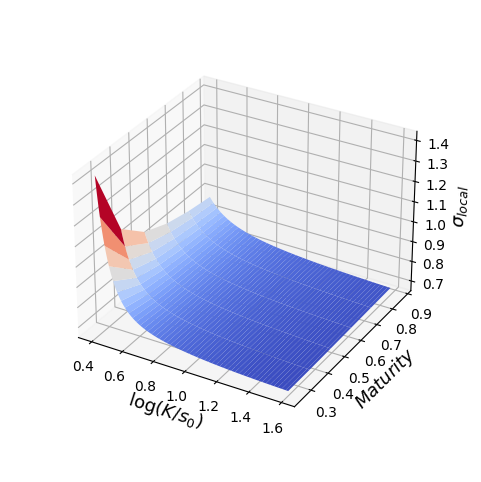}
\caption{The local volatility surface computed by GAN-1 with soft constraints (left), GAN-2 with soft constraints (middle) and FDM (right).}
\label{fig:comp_models_loc_vol}
\end{figure}

To compare our models quantitatively our generated local volatility was used to reprice the Heston option price and a error heatmap of ARPE, MRPE and the standard deviation of relative error as shown in figure \ref{fig:comp_models_loc_vol_err}. These heatmaps were generated by taking the relative error at for each maturity and strike. At each maturity and strike we computed the ARPE, MRPE and standard deviation for different sets of model parameters. From figure \ref{fig:comp_models_loc_vol_err} we can see that GAN-1 has a maximum ARPE of $0.1\%\pm0.175\%$ and a MRPE of $1.2\%$. GAN-2 has a maximum ARPE of $0.175\%\pm 0.175\%$ and a MRPE of $1\%$. We also note that the training time of GAN-1 $230.9626$ seconds and GAN-2 was $263.9366$ seconds. This opens up the possibility of using either models depending on the time constraints of the problem. The ARPE  of the FDM local volatility surface is $0.8\%\pm0.175\%$ and the MRPE is $14\%$.

\begin{figure}[htbp]
\centering
\includegraphics[width=\textwidth]{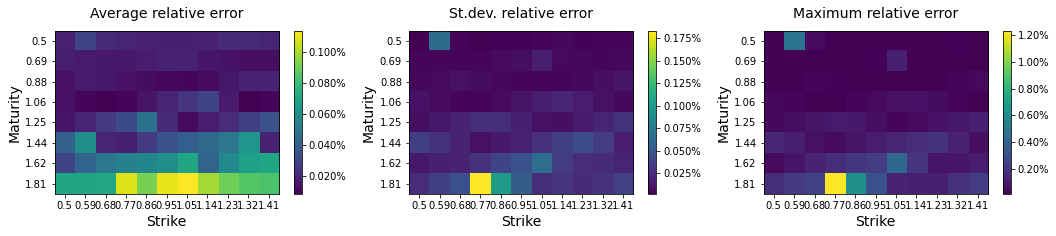}
\includegraphics[width=\textwidth]{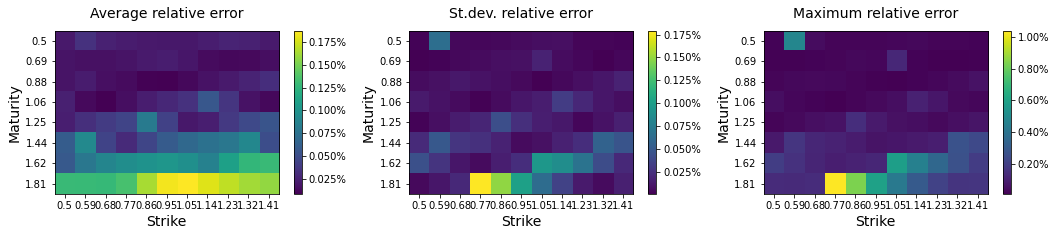}
\includegraphics[width=\textwidth]{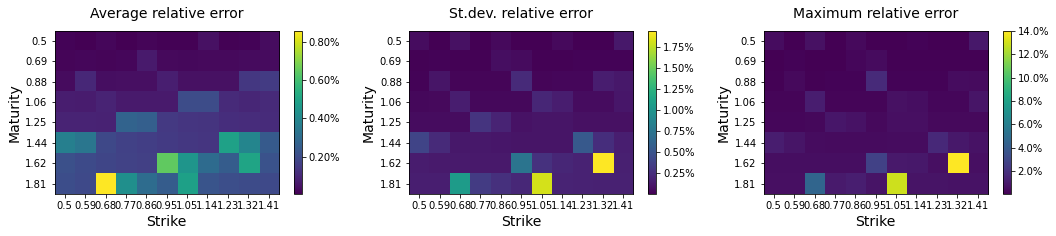}
\caption{Pricing error heatmap computed from local volatility.  GAN-1 with soft constraints (top), GAN-2 with soft constraints(middle) and FDM (bottom). The warmer colours (yellow) represent higher error values and colder colours (blue) represent lower error values.}
\label{fig:comp_models_loc_vol_err}
\end{figure}

In a similar fashion to the implied volatility, we also look at the arbitrage violations of our proposed methods for the local volatility. However, in contrast to the results of the implied volatility surface, the no-arbitrage soft constraints did not aid our proposed method in avoiding arbitrage violations significantly as seen in tables \ref{tab:butterfly_arbitrage_violations_loc_vol} and \ref{tab:calendar_arbitrage_violation_loc_vols}. We see that with or without soft constraints the generated local volatility does not violate butterfly arbitrage and effectively does not violate the calendar arbitrage in testing. However, this does not mean the soft constraints can be removed from training the generator. 

We look at the cross section of the local volatility surface at each maturity to see highlight the effects of soft constraints in generated local volatility surfaces as shown in figure \ref{fig:gan2_constraintsvnone}. The soft constraints influence the geometry of the generated local volatility surface. The effect on GAN-2 is the most prominent. Qualitatively, we see that there is a level of concavity when we generate the local volatility without soft constraints in the log moneyness region $[-0.3,0.3]$ at higher maturities. 

\begin{table}[htbp]
\centering
\caption{Number of butterfly arbitrage violations detected in local volatility surface.}
\begin{tabular}{lcccc}
\hline
\multicolumn{3}{c}{\textbf{GAN-1}} & \multicolumn{2}{c}{\textbf{GAN-2}} \\
\hline
& Soft Constraint & No Constraint & Soft Constraint & No Constraint \\
\hline
Training & 0/28306 & 0/28306 & 0/28306 & 0/28306 \\
Testing & 0/4995 & 0/4995 & 0/4995 & 0/4995 \\
\hline
\end{tabular}
\label{tab:butterfly_arbitrage_violations_loc_vol}
\end{table}

\begin{table}[htbp]
\centering
\caption{Number of calendar arbitrage violations detected in local volatility surface.}
\begin{tabular}{lcccc}
\hline
\multicolumn{3}{c}{\textbf{GAN-1}} & \multicolumn{2}{c}{\textbf{GAN-2}} \\
\hline
& Soft Constraint & No Constraint & Soft Constraint & No Constraint \\
\hline
Training & 0/28306 & 28/28306 & 0/28306 & 0/28306 \\
Testing & 0/4995 & 1/4995 & 0/4995 & 0/4995 \\
\hline
\end{tabular}
\label{tab:calendar_arbitrage_violation_loc_vols}
\end{table}

\begin{figure}[htbp]
\centering
\includegraphics[width=\textwidth]{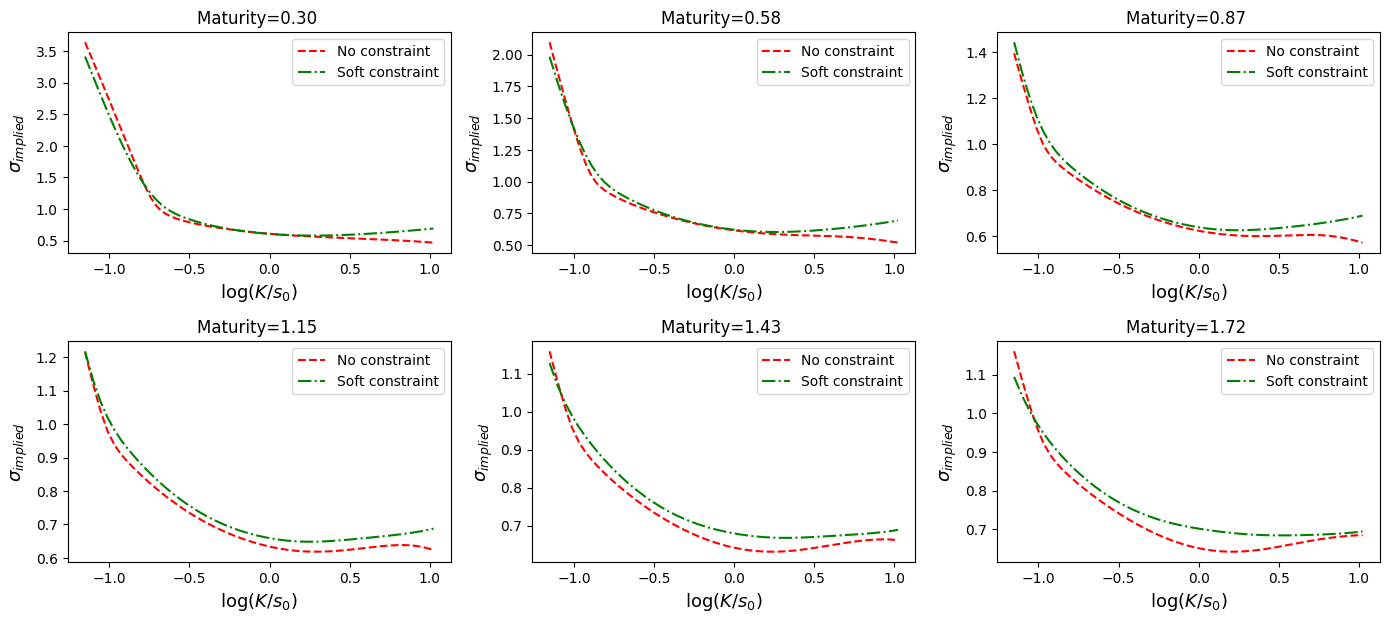}
\caption{The local volatility computed by GAN-2 with no constraints (red) and GAN-2 with soft constraints (green).}
\label{fig:gan2_constraintsvnone}
\end{figure}

\subsubsection{Comparison to GAN without MSE Loss}

We construct GAN-1 and GAN-2 without MSE loss function to compare our formulation with the formulation presented in \citep{sidogi-2022}. In our experiment we found that without the $MSE$ loss term the GAN is unable to fully capture the implied volatility surface. More specifically, it undershoots the volatility value resulting in a surface that is inconsistent with the market. This is likely due to the synthetic data being difficult to segment. In figure \ref{fig:comp_loss_models_fit} we show the GAN-1 and GAN-2 network without $L_{MSE}$. 

\begin{figure}[htbp]
\centering
\includegraphics[width=\textwidth]{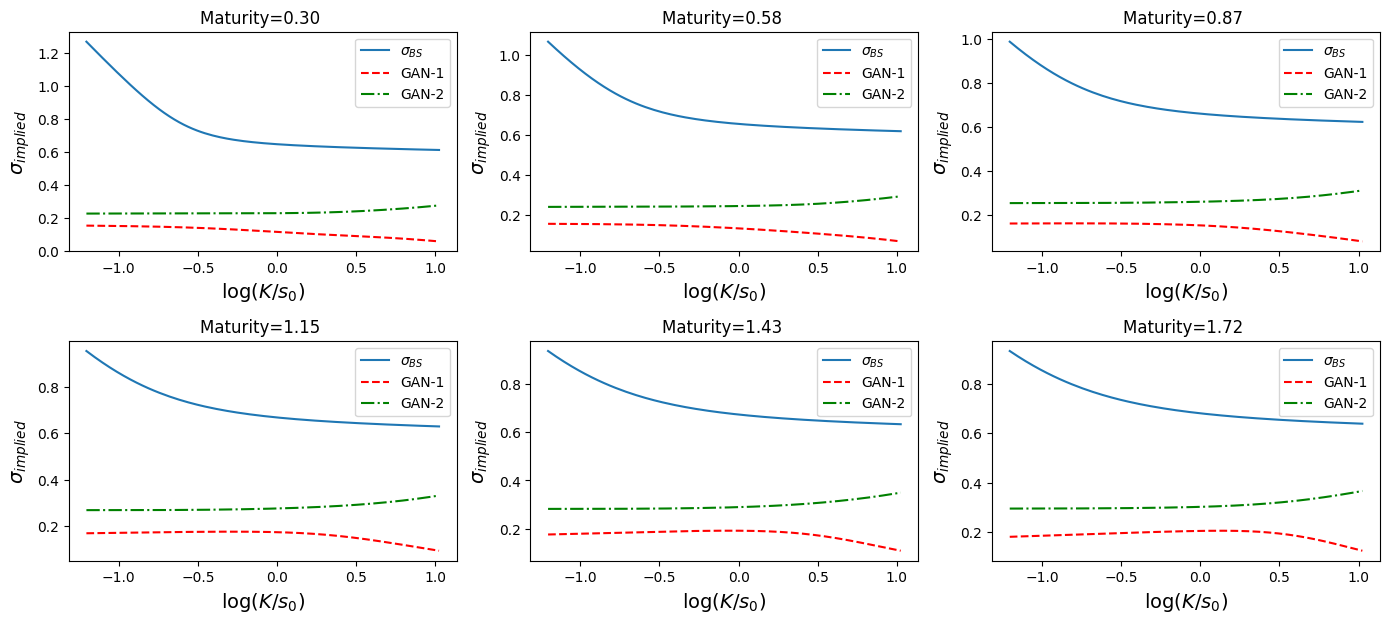}
\caption{The implied volatility computed by GAN-1 without MSE (red) and GAN-2 without MSE 
 (green) compared to the Black-Scholes implied volatility (blue).}
\label{fig:comp_loss_models_fit}
\end{figure}

In figures \ref{fig:comp_loss_models_fit}, we see that the implied volatility generated by the GAN without $MSE$ cannot fit the ground truth implied volatility at any point and fails to estimate the implied volatility surface that is market consistent. This highlights the importance in the inclusion of the MSE term in the loss function.

\subsection{Experiment 2: Comparison of Implied Volatility with IV-ANN}
In this experiment, we compare the performance of our proposed methods to the IV-ANN method \citep{liu-2019-1} with $4$ hidden layers with the ReLU activation function. We design the GAN network such that the MAE of the two approaches are similar. Doing this we find that our proposed GAN-2 is roughly two times faster than the IV-ANN approach. The performance is gauged by the mean absolute error of the predicted values, the MAE for the IV-ANN was reported as $9.73e^{-5}$ in \citep{liu-2019-1}. We summarized the training time, mean absolute error and the mean absolute percent error for each method in table \ref{tab:paper1}. As seen in table \ref{tab:paper1}, using the same training data and testing set we can see that GAN-2 outperforms IV-ANN in runtime when they are both trained to a similar mean absolute error.

\begin{table}[htbp]
\caption{Timing and error comparison between GAN-2 and IV-ANN}
\centering
\begin{tabular}{lccc}
\hline
\textbf{Model} & \textbf{Training Time} & \textbf{MAE} &\textbf{MAPE}\\
\hline
IV-ANN & 446.204831 & $2.3235e^{-5}$ &$0.044106\%$ \\
GAN-2 & 209.194824 & $2.1376e^{-5}$ & $0.039810\%$\\
\hline
\end{tabular}
\label{tab:paper1}
\end{table}

Next we evaluate the quality of our proposed method vs the IV-ANN method. We show the implied volatility learned by the IV-ANN method in figure \ref{fig:iv_ann_imp_vol}. Figure \ref{fig:iv_ann_imp_vol} looks at the cross section of the implied volatility surface generated by  the IV-ANN method for different maturities. From figure \ref{fig:iv_ann_imp_vol} we see that the implied volatility learned by IV-ANN has difficulty fitting the Black-Scholes implied volatility qualitatively.

\begin{figure}[htbp]
\centering
\includegraphics[width=\textwidth]{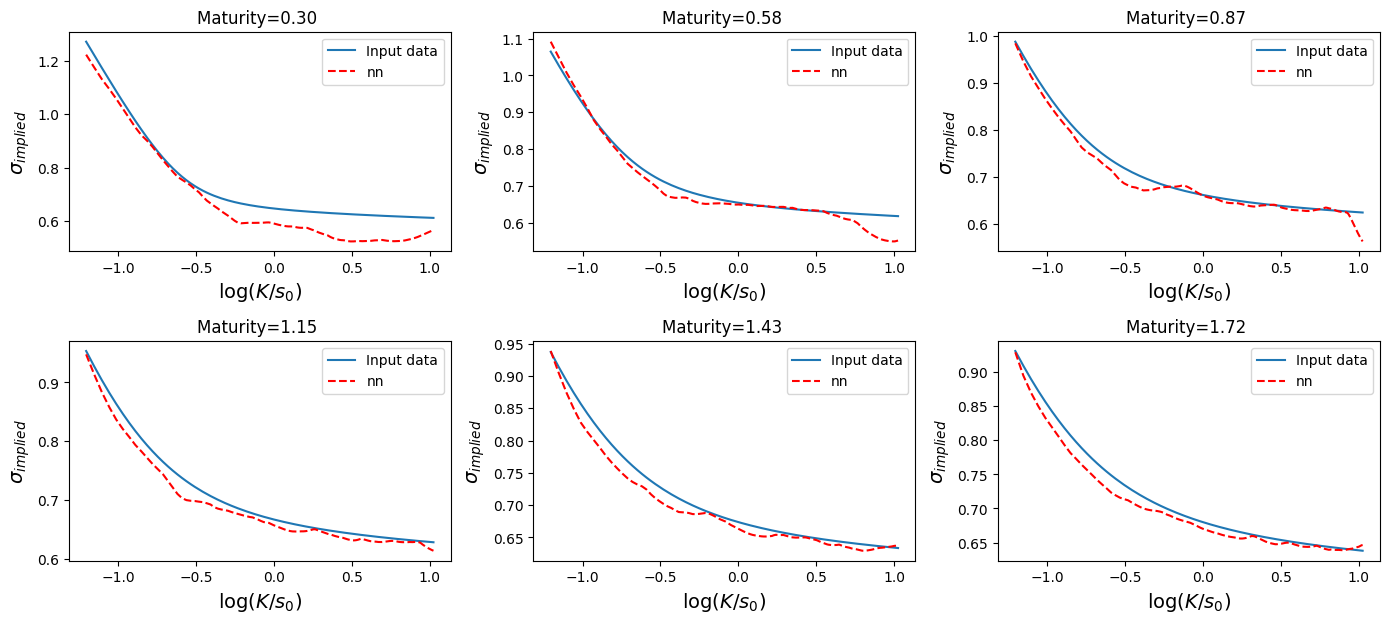}
\caption{Implied volatility cross section generated by the IV-ANN for different maturities.}
\label{fig:iv_ann_imp_vol}
\end{figure}

We compare the repricing performance of our proposed method vs the IV-ANN method. The repricing error heatmap of the IV-ANN method and GAN-2 is show in figure \ref{fig:paper3}. We see that our proposed GAN-2 network outperforms the IV-ANN method. The IV-ANN method has a maximum ARPE of $0.080\%\pm 0.12\%$ and a MRPE of $0.70\%$. Our proposed method has a maximum ARPE of $0.08\%\pm 0.1\%$, and a MRPE of $0.50\%$.

\begin{figure}[htbp]
\centering
\includegraphics[width=\textwidth]{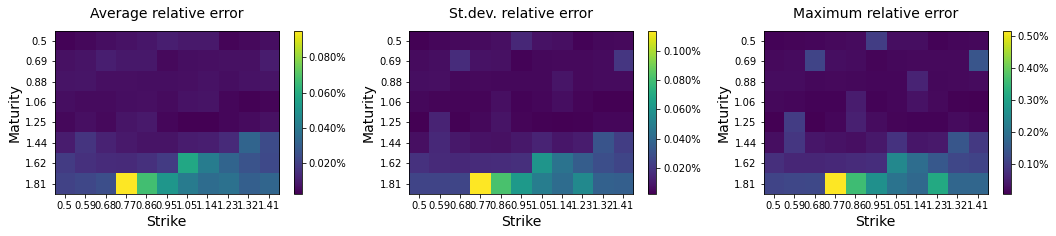}
\includegraphics[width=\textwidth]{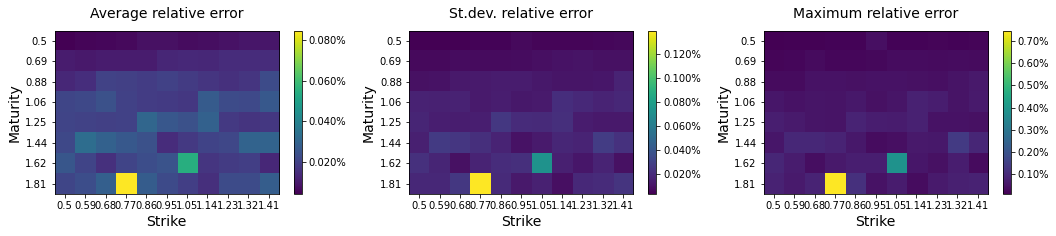}
\caption{GAN-2 implied volatility (top) and IV-ANN implied volatility (bottom repricing error heatmap. The warmer colors (yellow) respresent higher error values and colder colors (blue) represent lower error values.}
\label{fig:paper3}
\end{figure}

\subsection{Experiment 3: Comparison of Local Volatility Computation with Deep Neural Network and SSVI}
In this experiment we compare the local volatility computed by our method with a deep neural network (DNN). We use the pricing error given by the SSVI method \citep{gatheral-2014} as a benchmark. We implemented the deep neural network method using $4$ hidden layers with $400$ nodes with ReLU activation functions. The SSVI was computed using the analytical equation \eqref{eq:heston_like_ssvi} with the Heston-like parametric function in equation \eqref{eq:heston-parametric-function}. Note that the local volatility surface is not unique \citep{lee-2002}, thus many different local volatility surfaces may satisfy given market prices.

In this experiment we trained our proposed method and the DNN using the data constructed according to section \ref{sec:data}. Then the local volatility was computed using our proposed method, the deep neural network, and SSVI based on out-of training data generated. The training time, maximum ARPE and MRPE for our proposed method, the DNN and the benchmark SSVI is summarized in table \ref{tab:paper3_model_performance}.

\begin{table}[htbp]
\centering
\caption{We compare the performance of GAN-1, GAN-2, and DNN to the SSVI.}
\begin{tabular}{lccc}
\hline
\textbf{Model} & \textbf{Training Time} & \textbf{Max ARPE} & \textbf{MRPE} \\
\hline
SSVI & -- & $0.8\%\pm 0.2\%$ & $1.0\%$\\
GAN-1 & 198.689732$s$ & $0.1\%\pm 0.175\%$ & $1.2\%$\\
GAN-2 & 208.098838$s$ & $0.175\%\pm 0.175\%$ & $1.0\%$\\
DNN & 246.32142$s$ & $0.8\%\pm 0.2\%$ & $3.5\%$\\
\hline
\end{tabular}
\label{tab:paper3_model_performance}
\end{table}

The results of table \ref{tab:paper3_model_performance} show that our method is more efficient and accurate than the DNN. We also compared the arbitrage violations that occurred in our proposed method to the DNN and SSVI method. The arbitrage violations for the DNN and SSVI are shown in table \ref{tab:arbitrage_violation_loc_vols_dlv_ssvi}. We can see that during training and testing all methods do not produce any arbitrage opportunities. 

\begin{table}[htbp]
\centering
\caption{Number of arbitrage violations detected in local volatility surface computed by the DNN and SSVI method.}
\begin{tabular}{lccc}
\hline
 & \textbf{GAN-2} & \textbf{DNN} & \textbf{SSVI}\\
\hline
& Soft Constraint & Soft Constraint & Hard Constraint\\
\hline
Training & 0/28306  & 0/28306  & 0/28306 \\
Testing & 0/4995 & 0/4995 & 0/4995\\
\hline
\end{tabular}
\label{tab:arbitrage_violation_loc_vols_dlv_ssvi}
\end{table}

Finally we compared the pricing error produced by each method is shown in figures \ref{fig:comp_models_loc_vol_err}, and \ref{fig:paper3_price_error_dlv_ssvi}. We can clearly see that our proposed method GAN-2 has a maximum ARPE of $0.175\%\pm 0.175\%$ . This is better than the SSVI and DNN with maximum ARPE of $0.8\%\pm 0.2\%$. Our proposed method GAN-2 also performs comparably to the SSVI in terms of MRPE of $1\%$ (For GAN-2) and better than the DNN with a MRPE of $3.5\%$.

\begin{figure}[htbp]
\centering
\includegraphics[width=\textwidth]{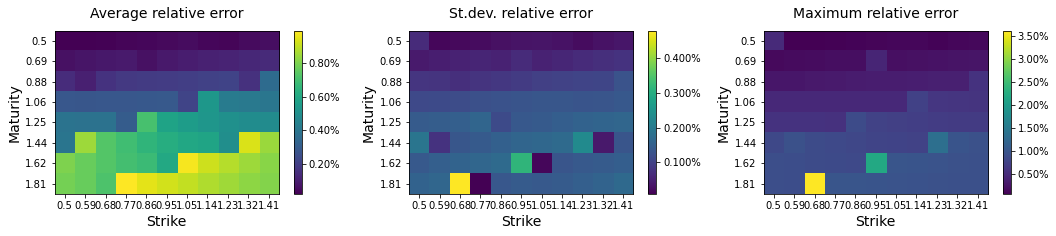}
\includegraphics[width=\textwidth]{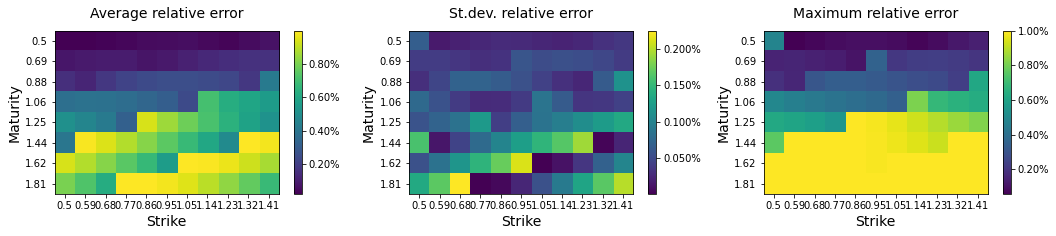}
\caption{Pricing error computed from local volatility using DNN (top) and SSVI (bottom). The warmer colours (yellow) represent higher error values and colder colours (blue) represent lower error values.}
\label{fig:paper3_price_error_dlv_ssvi}
\end{figure}

\subsection{Experiment 4: Comparison of Local Volatility Computation with VAE}

In this experiment we compare the local volatility computed by our proposed method with and a VAE . We implemented the VAE using $2$ hidden layers with $200$ nodes. The number of hidden layers and nodes were chosen such that the training time for the VAE and our GAN methods were the same. The encoder was constructed with a ReLU activation function for the hidden layers and a sigmoid activation function for the output. The decoder was constructed with ReLU activation function for all layers. The VAE was trained over the same training set as our proposed method as outlined in experiment 3. To compare our methods we use the same set up as in experiment 3. Given similar training times of $229.9448$ seconds, the pricing errors for the VAE trained is shown in figure \ref{fig:vae_error}. The maximum ARPE for VAE method is $2\% \pm 0.8\%$ and the MRPE is  $50\%$ this is due to the noise generated in VAE models. Given the same training time, the VAE performs worse than the results of our proposed GAN-1 as we saw in experiment 1.

\begin{figure}[htbp]
\centering
\includegraphics[width=0.32\textwidth]{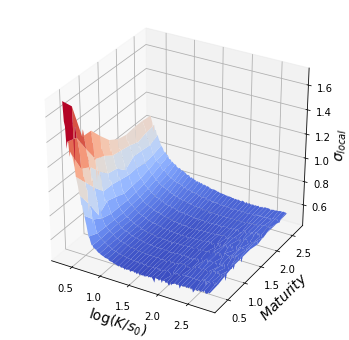}
\includegraphics[width=0.32\textwidth]{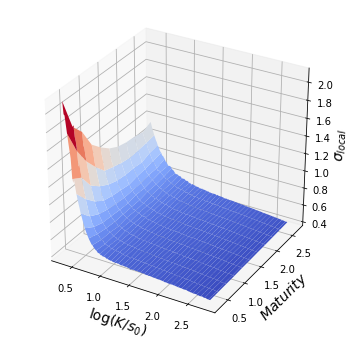}
\caption{Local volatility surface generated by trained VAE using a shallow network (left) and a deep network (right).}
\label{fig:vae_loc_vol}
\end{figure}

Give the same training time, the local volatility surface by our proposed method is smooth as shown in figure \ref{fig:comp_models_loc_vol}. This is in contrast to the noisy surface generated from VAEs as shown in the left hand side of figure \ref{fig:vae_loc_vol}. We note that the VAE method can produce smooth local volatility surfaces as shown in the right hand side of figure \ref{fig:vae_loc_vol} however it requires a deeper network for the encoder or larger batch sizes. Adding the extra layers  to smoothen the local volatility surface resulted in a training time of $1036.027$ seconds. Noise is present in VAE because the encoder network is trained to generate the parameters of a normal distribution that fits the data and data can be sampled using these learned parameters \citep{kingma-2014-vae}. This adds noise to the output if there is not a sufficient amount of samples generated.

\begin{figure}[htbp]
\centering
\includegraphics[width=\textwidth]{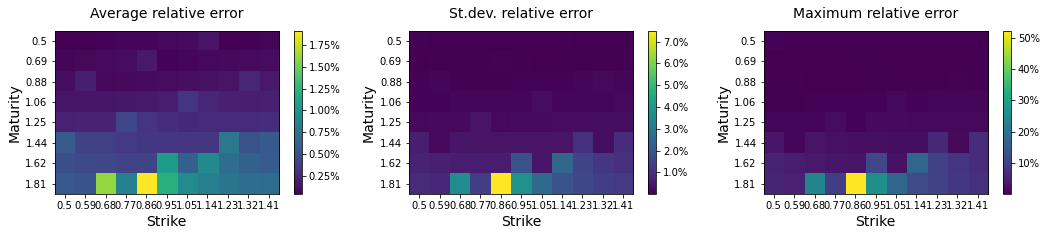}
\caption{Pricing errors computed for different maturities and strike using the VAE method. The warmer colours (yellow) represent higher error values and colder colours (blue) represent lower error values.}
\label{fig:vae_error}
\end{figure}

\subsection{Experiment 5: Generating Market Consistent Volatility Surfaces}
In this experiment we generate market volatility surfaces using limited market data. This experiment shows that our method can be adapted to be used on market data with minimal additional training with limited sample data. This is achieved using a pre-trained a 2-layer discriminator with synthetic data. We fine-tune the 2-layer generator which is retrained on a limited subset of the market data over $200$ epochs with a batch size of $32$. The GAN-2 model is tested on a separate test set that was not used for training the generator.

\subsubsection{Implied Volatility Smile}
We use GAN-2 to generate a market consistent implied volatility surface, this shows the capability of our network to generate market consistent volatility surfaces just by retraining the generator. We present the market consistent implied volatility surface in figure \ref{fig:mkt_imp_vol}.

\begin{figure}[htbp]
\centering
\includegraphics[width=0.32\textwidth]{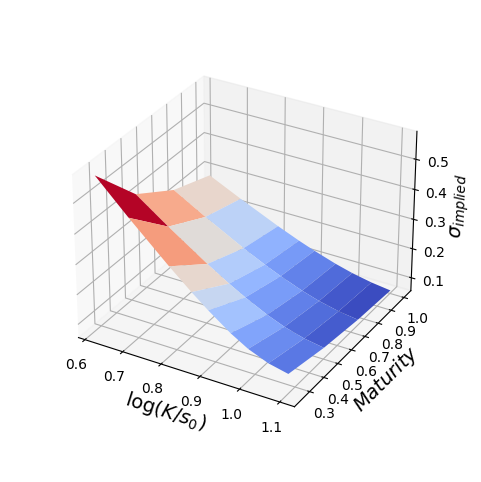}
\includegraphics[width=0.32\textwidth]{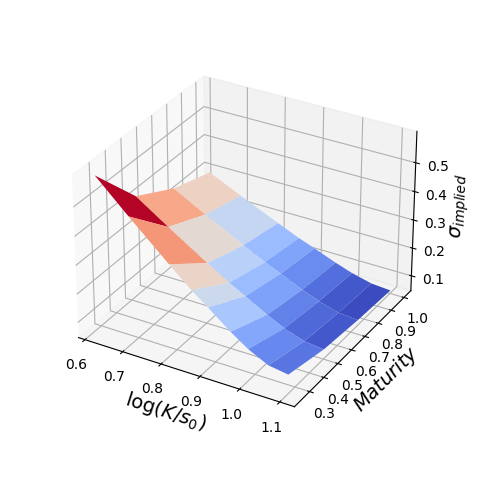}
\caption{Implied volatility generated by the generator (left) and the market implied volatility (right).}
\label{fig:mkt_imp_vol}
\end{figure}

Figure \ref{fig:mkt_imp_vol} shows the implied volatility surfaces from our generator and the market implied volatility. We see they are near identical qualitatively. Quantitatively, we measured the MAE score for GAN-2 as $0.007$ and the MAPE score as $7.28\%$.

\subsubsection{Local Volatility Smile}
We use GAN-2 to generate a market consistent local volatility surface, the market local volatility was approximated using the finite difference method. We present the market consistent local volatility surface in figure \ref{fig:mkt_loc_vol}.

\begin{figure}[htbp]
\centering
\includegraphics[width=0.32\textwidth]{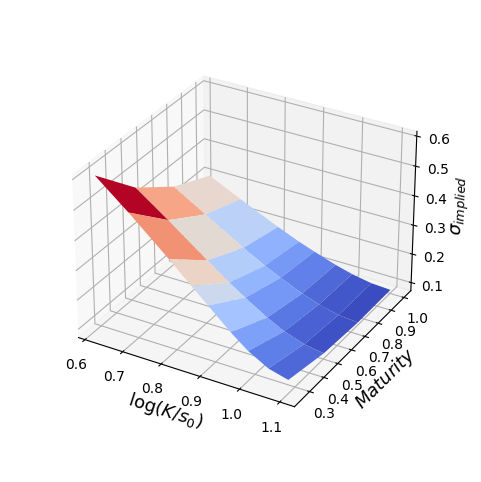}
\includegraphics[width=0.32\textwidth]{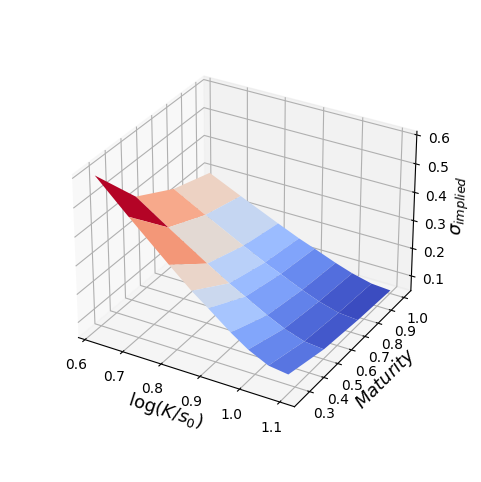}
\caption{Local volatility generated by the generator (left) and the market implied volatility (right).}
\label{fig:mkt_loc_vol}
\end{figure}

Figure \ref{fig:mkt_loc_vol} shows the local volatility surfaces from our generator and the market implied volatility. Quantitatively, we look at the repricing errors of from our local volatility given in figure \ref{fig:mkt_error}.

\begin{figure}[htbp]
\centering
\includegraphics[width=\textwidth]{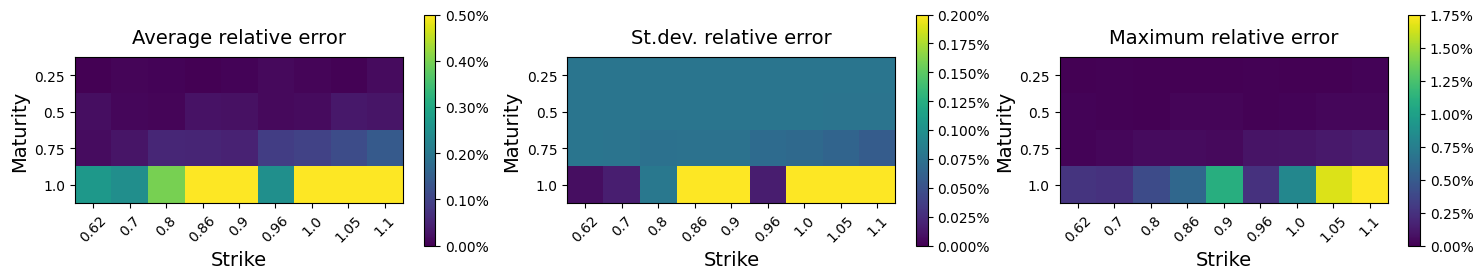}
\caption{Pricing errors computed for different maturities and strike using GAN-2 for market local volatilities. The warmer colours (yellow) represent higher error values and colder colours (blue) represent lower error values.}
\label{fig:mkt_error}
\end{figure}

In figure \ref{fig:mkt_error} we see smaller errors in shorter time-to-maturity over longer time-to-maturity at strikes that are far OTM. This is due to the lack of data points in this region. 

\section{Conclusion}\label{sec:Conclusion}

In this paper, we present a framework to generate volatility surfaces efficiently using GAN. By using a generator and discriminator together we are able to create a model framework that is efficient and accurate with shallow-narrow networks. We used no-arbitrage penalty terms with the MSE loss function to penalize arbitrage opportunities generated by the generators. The log-likelihood estimation of the discriminator adds to the generator by allowing the posterior generator to successfully output a valid volatility surface that is consistent with the option price. The discriminator was trained as a classifier to classify the volatility as true or false.

Our numerical results show that GAN-2 is outperforms GAN-1. However, a deeper discriminator network doesn't generate performance better than a shallow one. The majority of arbitrage violations detected by our discriminator is driven by butterfly arbitrage. Regularization using no-arbitrage soft constraints helps mitigate this for implied volatility computation. We show that our method is more accurate than the IV-ANN method with a faster training time and more accurate predictions. In particular, the GAN-2 model only required $209.195$ seconds and produced a MAPE of $3.981e^{-5}$, compared to our implementation of the IV-ANN model which took $446.205$ seconds with a MAPE of $4.4106e^{-5}$.

We further showed the capability of our network to compute the local volatility and compared our model with a deep neural network and SSVI. As shown in figure \ref{fig:paper3_price_error} our method produces a maximum ARPE of $0.1\%\pm0.117\%$ which is better than both the deep neural network and SSVI method with maximum ARPE of $0.8\%\pm0.2\%$. Our proposed method also has a MRPE comparable to the SSVI at $1\%$.  We have shown that our model can be generalized to market data with limited data samples using a pre-trained discriminator and fine-tuning the generator. We show that we can generate market consistent volatility surfaces in figures \ref{fig:mkt_imp_vol} and \ref{fig:mkt_loc_vol}.

With this paper we have shown the potential of generative methods in solving the inverse problem. In this work we assume that only one price exists for an option. It would be interesting to explore generative models behave in a market setting with more than one offered price and how it learns the volatility surfaces. It would also be interesting to see an extension into diffusion models to improve generated surface quality.

\section*{Disclosure Statement}
No potential conflict of interest is reported by the author(s).

\section*{Funding}
This work was supported by the Natural Sciences and Engineering Research Council of Canada.

\section*{ORCID}
\begin{description}
\item \textit{Andrew Na}: https://orcid.org/
0000-0002-6162-8171
\item \textit{Justin W.L. Wan}: https://orcid.org/
0000-0001-8367-6337
\end{description}

\bibliographystyle{plainnat}
\bibliography{References}

\end{document}